\documentclass[a4paper]{article}
\def\nodraft {}
\pdfoutput=1

\usepackage[a4paper]{geometry}
\usepackage{amsmath}
\usepackage{amssymb}
\usepackage{subcaption}

\usepackage{amsthm}

\usepackage{tikz}
\usetikzlibrary{topaths,calc}
\usetikzlibrary{positioning,shapes.geometric}
\usetikzlibrary{decorations.markings}
\usetikzlibrary{decorations.pathmorphing}
\usetikzlibrary{patterns}

\newtheorem{thm}{Theorem}

\newtheorem{crll}{Corollary}

\theoremstyle{definition}
\newtheorem{defn}{Definition}
\newtheorem{prop}{Proposition}
\newtheorem{expl}{Example}
\newtheorem{rem}{Remark}

\DeclareMathOperator{\res}{res}
\DeclareMathOperator{\id}{id}
\DeclareMathOperator{\im}{im}

\newcommand{\meet}{\wedge}
\newcommand{\join}{\vee}

\newcommand{\hopfrt}{\mathcal{H}^\text{\normalfont{rt}}}
\newcommand{\hopffg}{\mathcal{H}^\text{\normalfont{fg}}}
\newcommand{\hopffgs}{\widetilde{\mathcal{H}}^\text{\normalfont{fg}}}
\newcommand{\hopfpos}{\mathcal{H}^\text{\normalfont{P}}}
\newcommand{\hopflat}{\mathcal{H}^\text{\normalfont{L}}}

\newcommand{\unit}{\text{\normalfont{u}}}
\newcommand{\counit}{\epsilon}

\newcommand{\be}{\begin{equation}}
\newcommand{\ee}{\end{equation}}
\newcommand{\bea}{\begin{eqnarray}}
\newcommand{\eea}{\end{eqnarray}}
\newcommand{\beas}{\begin{eqnarray*}}
\newcommand{\eeas}{\end{eqnarray*}}

\newcommand{\subdiags}{\mathcal{P}}
\newcommand{\sdsubdiags}{\mathcal{P}^\text{s.d.}}
\newcommand{\sdsubdiagsn}{\widetilde{\mathcal{P}}^\text{s.d.}}

\def\One{\mathbb{I}}

\def\id{\mathrm{id}}

\def\res{\mathrm{res}}

\pgfdeclaredecoration{complete sines}{initial}
{
    \state{initial}[
        width=+0pt,
        next state=sine,
        persistent precomputation={\pgfmathsetmacro\matchinglength{
            \pgfdecoratedinputsegmentlength / int(\pgfdecoratedinputsegmentlength/\pgfdecorationsegmentlength)}
            \setlength{\pgfdecorationsegmentlength}{\matchinglength pt}
        }] {}
    \state{sine}[width=\pgfdecorationsegmentlength]{
        \pgfpathsine{\pgfpoint{0.25\pgfdecorationsegmentlength}{0.5\pgfdecorationsegmentamplitude}}
        \pgfpathcosine{\pgfpoint{0.25\pgfdecorationsegmentlength}{-0.5\pgfdecorationsegmentamplitude}}
        \pgfpathsine{\pgfpoint{0.25\pgfdecorationsegmentlength}{-0.5\pgfdecorationsegmentamplitude}}
        \pgfpathcosine{\pgfpoint{0.25\pgfdecorationsegmentlength}{0.5\pgfdecorationsegmentamplitude}}
}
    \state{final}{}
}

\tikzset{
    photon/.style={
        decoration={complete sines, amplitude=0.15cm, segment length=0.2cm},
        decorate    
    },
    fermion/.style={
        decoration={
            markings,
            mark=at position 0.5 with {\node[transform shape, xshift=-0.5mm, fill=black, inner sep=1pt, draw, isosceles triangle]{};}
        },
        postaction=decorate
    },
    gluon/.style={
        decoration={coil, aspect=0.75, mirror, amplitude=.4mm, segment length=.8mm},
        decorate
    }, 
    meson/.style={
        dashed
    }, 
    left/.style={
        bend left=90,
        looseness=1.75
    }, 
    leftsoft/.style={
        bend left=45,
        looseness=1.25
    }
}

\tikzset{every loop/.style={looseness=12,min distance=1cm}}

\title{Feynman diagrams and their algebraic lattices}
\author{Michael Borinsky\\Math.\ \& Physics Dept.\\ Humboldt U.\\10099 Berlin\\Germany\and Dirk Kreimer\thanks{DK thanks the Alexander von Humboldt Foundation and the BMBF for support by an Alexander von Humboldt Professorship. It is a pleasure for both authors to thank David Broadhurst, Spencer Bloch, Dominique Manchon and Karen Yeats for helpful discussions.  We also thank Fr\'ederic Fauvet for organizing the workshop {\em Resurgence, Physics and Numbers}, Centro de Giorgi, Pisa, May 2015, and hospitality there.  }\\Math.\ \&  Physics Dept.\\ Humboldt U.\\10099 Berlin\\Germany}
\date{}

\begin{document}
\maketitle
\begin{abstract}
We present the lattice structure of Feynman diagram renormalization in physical QFTs from the viewpoint of Dyson--Schwinger--Equations and the core Hopf algebra of Feynman diagrams. The lattice structure encapsules the nestedness of diagrams. This structure can be used to give explicit expressions for the counterterms in zero-dimensional QFTs using the lattice-Moebius function. Different applications for the tadpole-free quotient, in which all appearing elements correspond to semimodular lattices, are discussed. 
\end{abstract}

\bibliographystyle{plain}
\section{The Hopf algebra of Feynman diagrams}
Following  \cite{ConnesKreimer2000} the BPHZ renormalization algorithm to obtain finite amplitudes in quantum field theory (QFT) shows that Feynman diagrams act as generators of a Hopf algebra $\hopffg$. Elaborate expositions of this Hopf algebra exist \cite{manchon2004}.

The coproduct of the Hopf algebra of Feynman diagrams on a renormalizable QFT takes the form
\begin{align}
    \label{eqn:def_cop}
        &\Delta: \Gamma \mapsto
        \sum \limits_{ \gamma \in \sdsubdiags_D(\Gamma) }
            \gamma \otimes \Gamma/\gamma& &:& &{\hopffg_D} \rightarrow {\hopffg_D} \otimes {\hopffg_D}, 
\end{align}
($\gamma=\emptyset$, $\gamma=\Gamma$ allowed) where $\Gamma/\gamma$ is the contracted diagram which is obtained by shrinking all edges of $\gamma$ in $\Gamma$ to a point and 
\begin{align}
\label{eqn:sdsubdiags}
  \sdsubdiags_D(\Gamma):= \left\{\gamma \subset \Gamma \text{ such that } 
      \gamma = \prod \limits_i \gamma_i, \gamma_i \in \subdiags_{\text{1PI}}(\Gamma) \text{ and } \omega_D( \gamma_i ) \leq 0 \right\},
\end{align}
is the set of \textit{superficially divergent subdiagrams} or s.d.\ subdiagrams. $\omega_D(\Gamma)$ denotes the power counting superficial degree of divergence of the diagram $\Gamma$ in $D$ dimensional spacetime in the sense of Weinberg's Theorem \cite{Weinberg1960}. These are subdiagrams of $\Gamma$ whose connected components are superficially divergent 1PI diagrams.

Applying an evaluation of graphs by renormalized Feynman rules $\Phi_R:\hopffg_D\to\mathbb{C}$, a specific Feynman diagram will always map to a unique power in $\hbar$, $\hbar^{h_1(\Gamma)}$. 

For renormalized Feynman rules the task is to produce for each graph providing  an unrenormalized integrand (a form on the de Rham side) and a domain of integration (the Betti side) a well-defined period by pairing those two sides.

There are two avenues to proceed to obtain renormalized Feynman rules: one can either introduce a regulator $\epsilon$ say (dimensional regularization being a prominent choice with spacetime dimension $D=4-\epsilon$) and work with unrenormalized Feynman rules $\Phi(\epsilon)$ depending on the regulator $\epsilon$, or one renormalizes the integrand first avoiding a regulator altogether.

In the former case, the pairing gives a Laurent series with poles of finite order in $\epsilon$. The degree of the pole is bounded by the coradical degree of the Feynman graph under consideration. Adding correction terms as dictated by the Hopf algebra provides an expression for which the regulator can be removed, $\epsilon\to 0$.

In the latter case, the integrand is relegated to correction terms -again dictated by the Hopf algebra- which amount to sequences of blow-ups 
with the length of the sequence bounded by the coradical degree \cite{BrKr,KrMadrid,BingenLect}. 

The nature of the coradical degree and its systematic study using the lattice structure of Feynman diagrams will be described in what follows. 

Using the reduced coproduct, $\widetilde{\Delta} = \Delta - \mathbb{I} \otimes \id - \id \otimes \mathbb{I}$, the coradical degree of an element $h \in \hopffg_D$ is the minimal number $d=: \mathrm{cor}(h)$ such that
\begin{align}
\underbrace{ (\id^{\otimes (d-1)} \otimes \widetilde{\Delta} ) \circ \cdots \circ (\id \otimes \widetilde{\Delta}) \circ \widetilde{\Delta}}_{d-\text{times}} h = 0.
\end{align}
The coradical degree of a Feynman diagram is a measure for the `nestedness' of Feynman diagrams. For instance, a Feynman diagram of coradical degree $1$ has no subdivergences. Such a diagram is a primitive element of the Hopf algebra $\hopffg_D$. A diagram with a single subdivergence has coradical degree $2$ and a diagram with a subdivergence, which has itself a subdivergence, coradical degree $3$ and so on. 

But what is the coradical degree if we have to deal with overlapping divergences? Of course, every diagram will have a well defined $\epsilon$ expansion even if it is not accessible by explicit calculation, but is there a combinatorial description that enables us to analyze the coradical degree directly? The answer can be found in the lattice structure of Feynman diagrams. 

\section{The lattice of subdiagrams}
It is obvious that $\sdsubdiags_D(\Gamma)$ is a poset ordered by inclusion. The statement that a subdiagram $\gamma_1$ covers $\gamma_2$ in $\sdsubdiags_D(\Gamma)$ is equivalent to the statement that $\gamma_1 / \gamma_2$ is primitive.

The Hasse diagram of a s.d.\ diagram $\Gamma$ can be constructed by the following procedure: 
\begin{enumerate}
\item
Draw the diagram and find all the maximal forests $\gamma_i \in  \sdsubdiags_D(\Gamma)$ such that $\Gamma/\gamma_i$ is primitive. 
\item Draw the diagrams $\gamma_i$ under $\Gamma$ and draw lines from $\Gamma$ to the $\gamma_i$. 
\item 
Subsequently, determine all the maximal forests $\mu_i$ of the $\gamma_i$ and draw them under the $\gamma_i$. 
\item 
Draw a line from $\gamma_i$ to $\mu_i$ if $\mu_i \subset \gamma_i$. 
\item Repeat this until only primitive diagrams are left. 
\item Then draw lines from the primitive subdiagrams to an additional $\emptyset$-diagram underneath them. 
\item In the end, replace diagrams by vertices.
\end{enumerate}

\begin{expl}
For instance, the set of superficially divergent subdiagrams for $D=4$ of the diagram,
\ifdefined\nodraft
$
{
\def\scale{2ex}
\begin{tikzpicture}[x=\scale,y=\scale,baseline={([yshift=-.5ex]current bounding box.center)}]
\begin{scope}[node distance=1]
\coordinate (v0);
\coordinate[right=.5 of v0] (v4);
\coordinate[above right= of v4] (v2);
\coordinate[below right= of v4] (v3);
\coordinate[below right= of v2] (v5);
\coordinate[right=.5 of v5] (v1);
\coordinate[above right= of v2] (o1);
\coordinate[below right= of v2] (o2);
\coordinate[below left=.5 of v0] (i1);
\coordinate[above left=.5 of v0] (i2);
\coordinate[below right=.5 of v1] (o1);
\coordinate[above right=.5 of v1] (o2);
\draw (v0) -- (i1);
\draw (v0) -- (i2);
\draw (v1) -- (o1);
\draw (v1) -- (o2);
\draw (v0) to[bend left=20] (v2);
\draw (v0) to[bend right=20] (v3);
\draw (v1) to[bend left=20] (v3);
\draw (v1) to[bend right=20] (v2);
\draw (v2) to[bend right=60] (v3);
\draw (v2) to[bend left=60] (v3);
\filldraw (v0) circle(1pt);
\filldraw (v1) circle(1pt);
\filldraw (v2) circle(1pt);
\filldraw (v3) circle(1pt);
\ifdefined\cvl
\draw[line width=1.5pt] (v0) to[bend left=20] (v2);
\draw[line width=1.5pt] (v0) to[bend right=20] (v3);
\fi
\ifdefined\cvr
\draw[line width=1.5pt] (v1) to[bend left=20] (v3);
\draw[line width=1.5pt] (v1) to[bend right=20] (v2);
\fi
\ifdefined\cvml
\draw[line width=1.5pt] (v2) to[bend left=60] (v3);
\fi
\ifdefined\cvmr
\draw[line width=1.5pt] (v2) to[bend right=60] (v3);
\fi
\end{scope}
\end{tikzpicture}
}$
can be represented as the Hasse diagram
$
{
\def\scale{2ex}
\begin{tikzpicture}[x=\scale,y=\scale,baseline={([yshift=-.5ex]current bounding box.center)}]
\begin{scope}[node distance=1]
\coordinate (top) ;
\coordinate [below left= of top] (v1);
\coordinate [below right= of top] (v2);
\coordinate [below left= of v2] (v3);
\coordinate [below= of v3] (bottom);
\draw (top) -- (v1);
\draw (top) -- (v2);
\draw (v1) -- (v3);
\draw (v2) -- (v3);
\draw (v3) -- (bottom);
\filldraw[fill=white, draw=black] (top) circle(2pt);
\filldraw[fill=white, draw=black] (v1) circle(2pt);
\filldraw[fill=white, draw=black] (v2) circle(2pt);
\filldraw[fill=white, draw=black] (v3) circle(2pt);
\filldraw[fill=white, draw=black] (bottom) circle(2pt);
\end{scope}
\end{tikzpicture}
}.
$
\else

MISSING IN DRAFT MODE

\fi
\end{expl}

The motivation to search for more properties of these posets came from the work of Berghoff \cite{berghoff2014wonderful}, who studied the posets of subdivergences in the context of Epstein-Glaser renormalization and discovered that the posets of diagrams with only logarithmic divergent subdivergences are distributive lattices.

An important observation to make is that the set of superficially divergent subdiagrams $\sdsubdiags_D(\Gamma)$ of a diagram $\Gamma$ is a lattice for a big class of QFTs. For convenience, we repeat the definition of a lattice here: 
\begin{defn}[Lattice]
A lattice is a poset $L$ for which an unique least upper bound (\textit{join}) and an unique greatest lower bound (\textit{meet}) exists for any combination of two elements in $L$. The join of two elements $x,y \in L$ is denoted as $x \vee y$ and the meet as $x \wedge y$.
Every lattice has a unique greatest element denoted as $\hat{1}$ and a unique smallest element $\hat{0}$. Every interval of a lattice is also a lattice.
\end{defn}

In many QFTs, $\sdsubdiags(\Gamma)$ is a lattice for every s.d.\ diagram $\Gamma$ \cite{borinsky2015}. The union of two subdiagrams will play the role of the meet.
\begin{defn}[Join-meet-renormalizable quantum field theory]
A renormalizable QFT is called join-meet-renormalizable if $\sdsubdiags_D(\Gamma)$, ordered by inclusion, is a lattice for every s.d.\ Feynman diagram $\Gamma$.
\end{defn}

It turns out to be a sufficient requirement on the set $\sdsubdiags_D(\Gamma)$ to be a lattice that it is closed under taking unions of subdiagrams. 
\begin{thm}
A renormalizable QFT is join-meet-renormalizable if $\sdsubdiags_D(\Gamma)$ is closed under taking unions: $\gamma_1, \gamma_2 \in \sdsubdiags_D(\Gamma) \Rightarrow \gamma_1 \cup \gamma_2 \in \sdsubdiags_D(\Gamma)$ for all s.d.\ diagrams $\Gamma$.
\end{thm}
\begin{proof}
$\sdsubdiags_D(\Gamma)$ is ordered by inclusion $\gamma_1 \leq \gamma_2 \Leftrightarrow \gamma_1 \subset \gamma_2$.
The join is given by taking the union of diagrams: $\gamma_1 \join \gamma_2 := \gamma_1 \cup \gamma_2$. $\sdsubdiags_D(\Gamma)$ has a unique greatest element $\hat{1} := \Gamma$ and a unique smallest element $\hat{0} := \emptyset$. Therefore $\sdsubdiags_D(\Gamma)$ is a lattice \cite[Prop. 3.3.1]{stanley1997}. The unique meet is given by the formula, $\gamma_1 \meet \gamma_2 := \bigcup \limits_{\mu \leq \gamma_1 \text{ and } \mu \leq \gamma_2} \mu$. 
\end{proof}

Not every Feynman diagram fulfills this requirement. A counterexample of a Feynman diagram of $\phi^6$-theory in $3$ dimensions where $\sdsubdiags_D(\Gamma)$ is not a lattice is given in figure \ref{fig:phi6nolattice}. The corresponding poset $\sdsubdiags_3(\Gamma)$ is depicted in figure \ref{fig:phi6nolatticeposet}.
{ 
\ifdefined\nodraft
\begin{figure}
  \subcaptionbox{Example of a diagram where $\sdsubdiags_3(\Gamma)$ is not a lattice.\label{fig:phi6nolattice}}
  [.45\linewidth]{
\def \scale{2em}
\begin{tikzpicture}[x=\scale,y=\scale,baseline={([yshift=-.5ex]current bounding box.center)}]
\begin{scope}[node distance=1]
\coordinate (v0);
\coordinate[right=of v0] (v4);
\coordinate[above right= of v4] (v2);
\coordinate[below right= of v4] (v3);
\coordinate[below right= of v2] (v5);
\coordinate[right=of v5] (v1);
\coordinate[above right= of v2] (o1);
\coordinate[below right= of v2] (o2);
\coordinate[below left=.5 of v0] (i1);
\coordinate[above left=.5 of v0] (i2);
\coordinate[below right=.5 of v1] (o1);
\coordinate[above right=.5 of v1] (o2);
\draw (v0) -- (i1);
\draw (v0) -- (i2);
\draw (v1) -- (o1);
\draw (v1) -- (o2);
\draw (v0) to[bend left=20] (v2);
\draw (v0) to[bend left=45] (v2);
\draw (v0) to[bend right=45] (v3);
\draw (v0) to[bend right=20] (v3);
\draw (v1) to[bend left=20] (v3);
\draw (v1) to[bend left=45] (v3);
\draw (v1) to[bend right=45] (v2);
\draw (v1) to[bend right=20] (v2);
\draw (v2) to[bend right=20] (v4);
\draw (v2) to[bend left=20] (v5);
\draw (v3) to[bend right=20] (v5);
\draw (v3) to[bend left=20] (v4);
\draw (v4) to[bend left=45] (v5);
\draw (v4) to[bend left=15] (v5);
\draw (v4) to[bend right=45] (v5);
\draw (v4) to[bend right=15] (v5);
\filldraw (v0) circle(1pt);
\filldraw (v1) circle(1pt);
\filldraw (v2) circle(1pt);
\filldraw (v3) circle(1pt);
\filldraw (v4) circle(1pt);
\filldraw (v5) circle(1pt);
\end{scope}
\end{tikzpicture}
  }
  \subcaptionbox{The corresponding non-lattice poset. Trivial vertex multiplicities were omitted. \label{fig:phi6nolatticeposet}}
  [.45\linewidth]{
\def \scale{1em}
\begin{tikzpicture}[x=\scale,y=\scale,baseline={([yshift=-.5ex]current bounding box.center)}]
\begin{scope}[node distance=1]
\coordinate (top) ;
\coordinate [below left=2 and 2 of top] (v1) ;
\coordinate [below left=2 and 1 of top] (v2) ;
\coordinate [below =2 of top] (v4) ;
\coordinate [below right= 2 and 1 of top] (v5) ;
\coordinate [below left = 4 and 1 of top] (w1) ;
\coordinate [below= 4 of top] (w2) ;
\coordinate [below left = 5 and .5 of top] (u1) ;
\coordinate [below right =3 and 3 of top] (u2) ;
\coordinate [below = 6 of top] (s) ;
\coordinate [below = 7 of top] (t) ;
\draw (top) -- (v1);
\draw (top) -- (v2);
\draw (top) -- (u2);
\draw (top) -- (v4);
\draw (top) -- (v5);
\draw[color=white,line width=4pt] (v4) -- (w1);
\draw (v4) -- (w1);
\draw[color=white,line width=4pt] (v2) -- (w2);
\draw (v2) -- (w1);
\draw (v2) -- (w2);
\draw (v4) -- (w2);
\draw (v1) -- (w1);
\draw (v5) -- (w2);
\draw[color=white,line width=4pt] (w1) -- (u1);
\draw[color=white,line width=4pt] (w2) -- (u1);
\draw (w1) -- (u1);
\draw (w2) -- (u1);
\draw (u1) -- (s);
\draw (u2) -- (s);
\draw (s) -- (t);
\filldraw[fill=white, draw=black] (top) circle(2pt);
\filldraw[fill=white, draw=black] (v1) circle(2pt);
\filldraw[fill=white, draw=black] (v2) circle(2pt);
\filldraw[fill=white, draw=black] (v4) circle(2pt);
\filldraw[fill=white, draw=black] (v5) circle(2pt);
\filldraw[fill=white, draw=black] (w1) circle(2pt);
\filldraw[fill=white, draw=black] (w2) circle(2pt);
\filldraw[fill=white, draw=black] (u1) circle(2pt);
\filldraw[fill=white, draw=black] (u2) circle(2pt);
\filldraw[fill=white, draw=black] (s) circle(2pt);
\filldraw[fill=white, draw=black] (t) circle(2pt);
\end{scope}
\end{tikzpicture}
  }
  \caption{Counter-example for a renormalizable but not join-meet-renormalizable QFT: $\phi^6$-theory in $3$ dimensions.}
  \label{fig:nolattice}
\end{figure}
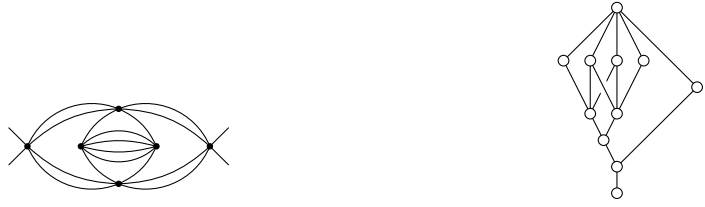
\else

MISSING IN DRAFT MODE

\fi
}

On the other hand, there is a large class of join-meet-renormalizable quantum field theories which includes the standard model as established by the following theorem:
\begin{thm}{\cite[Corr. 2]{borinsky2015}}
All renormalizable QFTs with only four-or-less-valent vertices are join-meet-renormalizable.
\end{thm}

This a surprising result. Lattices are very well studied objects in combinatorics. It is worthwhile to search for more properties which the lattices in physical QFTs carry. But first, we will look how the Hopf algebra and the lattice structure fit together.

\section{The Hopf algebra of decorated lattices}
It is well known that lattices and posets can be equipped with Hopf algebra structures \cite{Schmitt1994}. The Hopf algebra structure applicable in the present case is the following decorated version of an incidence Hopf algebra:
\begin{defn}[Hopf algebra of decorated posets]
Let $\mathcal{D}$ be the set of tuples $(P, \nu)$, where $P$ is a finite poset with a unique lower bound $\hat{0}$ and a unique upper bound $\hat{1}$ and a strictly order preserving map $\nu: P \rightarrow \mathbb{N}_0$ with $\nu(\hat{0}) = 0$. 
One can think of $\mathcal{D}$ as the set of bounded posets augmented by a strictly order preserving decoration.
An equivalence relation is set up on $\mathcal{D}$ by relating $(P_1, \nu_1) \sim (P_2, \nu_2)$ if there is an isomorphism $j: P_1 \rightarrow P_2$, which respects the decoration $\nu$: $\nu_1 = \nu_2 \circ j$. 

Let $\hopfpos$ be the $\mathbb{Q}$-algebra generated by all the elements in the quotient $\mathcal{P}/\sim$ with the commutative multiplication: 
\begin{align}
&m_{\hopfpos}: & 
\hopfpos &\otimes \hopfpos & 
&\rightarrow&
&\hopfpos, \\
&&  
(P_1, \nu_1) &\otimes (P_2, \nu_2) & 
&\mapsto& 
&\left( P_1 \times P_2, \nu_1 + \nu_2 \right),
\end{align}
which takes the Cartesian product of the two posets and adds the decorations $\nu$. The sum of the two functions $\nu_1$ and $\nu_2$ is to be interpreted in the sense: $(\nu_1 + \nu_2)(x,y) = \nu_1(x) + \nu_2(y)$. The singleton poset $P=\left\{\hat{0}\right\}$ with $\hat{0}=\hat{1}$ and the trivial decoration $\nu(\hat{0}) = 0$ serves as a multiplicative unit: $\unit(1) = \mathbb{I}_{\hopfpos} := (\left\{\hat{0}\right\}, \hat{0} \mapsto 0)$.

Equipped with the coproduct,
\begin{align}
    &\Delta_{\hopfpos}: &
    &\hopfpos& 
    &\rightarrow& 
    \hopfpos &\otimes \hopfpos, \\
\label{eqn:poset_cop}
  && &(P, \nu)& 
  &\mapsto&
  \sum \limits_{ x \in P } 
    ( [ \hat{0}, x ], \nu ) &\otimes \left( [x, \hat{1}], \nu - \nu(x) \right),
\end{align}
where $(\nu - \nu(x))(y) = \nu(y) - \nu(x)$ 
and the counit $\counit$ which vanishes on every generator except $\mathbb{I}_{\hopfpos}$, the algebra $\hopfpos$ becomes a counital coalgebra. 
\end{defn}
This algebra and coalgebra is in fact a Hopf algebra \cite{borinsky2015} which augments the corresponding incidence Hopf algebra by a decoration. The decoration is needed to capture at least the simplest invariant of a diagram: The loop number.

Having defined the Hopf algebra, we can setup a Hopf algebra morphism from $\hopffg_D$ to $\hopflat$:
\begin{thm}{\cite[Thm. 3]{borinsky2015}}
Let $\nu(\gamma) = h_1(\gamma)$.
\label{thm:hopf_alg_morph}
The map,
\begin{align}
    &\chi_D:& &\hopffg_D& &\rightarrow& &\hopfpos, \\
      &&       &\Gamma& &\mapsto& &( \sdsubdiags_D(\Gamma), \nu ),
\end{align}
which assigns to every diagram, its poset of s.d.\ subdiagrams decorated by the loop number of the subdiagram, is a Hopf algebra morphism.
\end{thm}

Because of the special structure of $\sdsubdiags_D(\Gamma)$ in join-meet-renormalizable theories, it follows immediately that:
\begin{crll}
\label{crll:joinmeethopflat}
In a join-meet-renormalizable QFT, $\im(\chi_D) \subset \hopflat \subset \hopfpos$, where $\hopflat$ is the subspace of $\hopfpos$ which is generated by all elements $(L, \nu)$, where $L$ is a lattice. In other words: In a join-meet-renormalizable QFT, $\chi_D$ maps s.d.\ diagrams and products of them to decorated lattices. 
\end{crll}

\begin{expl}
For any primitive 1PI diagram, i.e.\ $\Gamma \in \ker \widetilde \Delta$,
\ifdefined\nodraft
\begin{align}
\chi_D( \Gamma ) = ( \sdsubdiags_D(\Gamma), \nu ) = 
{\def \scale {4ex}
\begin{tikzpicture}[x=\scale,y=\scale,baseline={([yshift=-.5ex]current bounding box.center)}]
\begin{scope}[node distance=1]
\coordinate (top) ;
\coordinate [below= of top] (bottom);
\draw (top) -- (bottom);
\filldraw[fill=white, draw=black,circle] (top) node[fill,circle,draw,inner sep=1pt]{$L$};
\filldraw[fill=white, draw=black,circle] (bottom) node[fill,circle,draw,inner sep=1pt]{$0$};
\end{scope}
\end{tikzpicture}
},
\end{align}
where the vertices in the Hasse diagram are decorated by the value of $\nu$ and 
$L = h_1(\Gamma)$ is the loop number of the primitive diagram. 

The coproduct of $\chi_D(\Gamma)$ in $\hopfpos$ can be calculated using eq. \ref{eqn:poset_cop}:
\begin{align}
\label{eqn:explcopposet}
\Delta_{\hopfpos} 
{\def \scale {4ex}
\begin{tikzpicture}[x=\scale,y=\scale,baseline={([yshift=-.5ex]current bounding box.center)}]
\begin{scope}[node distance=1]
\coordinate (top) ;
\coordinate [below= of top] (bottom);
\draw (top) -- (bottom);
\filldraw[fill=white, draw=black,circle] (top) node[fill,circle,draw,inner sep=1pt]{$L$};
\filldraw[fill=white, draw=black,circle] (bottom) node[fill,circle,draw,inner sep=1pt]{$0$};
\end{scope}
\end{tikzpicture}
} = 
{\def \scale {4ex}
\begin{tikzpicture}[x=\scale,y=\scale,baseline={([yshift=-.5ex]current bounding box.center)}]
\begin{scope}[node distance=1]
\coordinate (top) ;
\coordinate [below= of top] (bottom);
\draw (top) -- (bottom);
\filldraw[fill=white, draw=black,circle] (top) node[fill,circle,draw,inner sep=1pt]{$L$};
\filldraw[fill=white, draw=black,circle] (bottom) node[fill,circle,draw,inner sep=1pt]{$0$};
\end{scope}
\end{tikzpicture}
}
\otimes 
\mathbb{I}
+
\mathbb{I}
\otimes
{\def \scale {4ex}
\begin{tikzpicture}[x=\scale,y=\scale,baseline={([yshift=-.5ex]current bounding box.center)}]
\begin{scope}[node distance=1]
\coordinate (top) ;
\coordinate [below= of top] (bottom);
\draw (top) -- (bottom);
\filldraw[fill=white, draw=black,circle] (top) node[fill,circle,draw,inner sep=1pt]{$L$};
\filldraw[fill=white, draw=black,circle] (bottom) node[fill,circle,draw,inner sep=1pt]{$0$};
\end{scope}
\end{tikzpicture}
}.
\end{align}
As expected, these decorated posets are also primitive in $\hopfpos$.
\else

MISSING IN DRAFT MODE

\fi
\end{expl}
\begin{expl}
For the diagram
\ifdefined\nodraft
$
{\def \scale {1.5ex}
\begin{tikzpicture}[x=\scale,y=\scale,baseline={([yshift=-.5ex]current bounding box.center)}]
\begin{scope}[node distance=1]
\coordinate (v0);
\coordinate[right=.5 of v0] (v4);
\coordinate[above right= of v4] (v2);
\coordinate[below right= of v4] (v3);
\coordinate[below right= of v2] (v5);
\coordinate[right=.5 of v5] (v1);
\coordinate[above right= of v2] (o1);
\coordinate[below right= of v2] (o2);
\coordinate[below left=.5 of v0] (i1);
\coordinate[above left=.5 of v0] (i2);
\coordinate[below right=.5 of v1] (o1);
\coordinate[above right=.5 of v1] (o2);
\draw (v0) -- (i1);
\draw (v0) -- (i2);
\draw (v1) -- (o1);
\draw (v1) -- (o2);
\draw (v0) to[bend left=20] (v2);
\draw (v0) to[bend right=20] (v3);
\draw (v1) to[bend left=20] (v3);
\draw (v1) to[bend right=20] (v2);
\draw (v2) to[bend right=60] (v3);
\draw (v2) to[bend left=60] (v3);
\filldraw (v0) circle(1pt);
\filldraw (v1) circle(1pt);
\filldraw (v2) circle(1pt);
\filldraw (v3) circle(1pt);
\ifdefined\cvl
\draw[line width=1.5pt] (v0) to[bend left=20] (v2);
\draw[line width=1.5pt] (v0) to[bend right=20] (v3);
\fi
\ifdefined\cvr
\draw[line width=1.5pt] (v1) to[bend left=20] (v3);
\draw[line width=1.5pt] (v1) to[bend right=20] (v2);
\fi
\ifdefined\cvml
\draw[line width=1.5pt] (v2) to[bend left=60] (v3);
\fi
\ifdefined\cvmr
\draw[line width=1.5pt] (v2) to[bend right=60] (v3);
\fi
\end{scope}
\end{tikzpicture}
} \in \hopffg_4$,
$\chi_D$ gives the decorated poset,
\begin{align}
\chi_D\left( 
{\def \scale {4ex}
\begin{tikzpicture}[x=\scale,y=\scale,baseline={([yshift=-.5ex]current bounding box.center)}]
\begin{scope}[node distance=1]
\coordinate (v0);
\coordinate[right=.5 of v0] (v4);
\coordinate[above right= of v4] (v2);
\coordinate[below right= of v4] (v3);
\coordinate[below right= of v2] (v5);
\coordinate[right=.5 of v5] (v1);
\coordinate[above right= of v2] (o1);
\coordinate[below right= of v2] (o2);
\coordinate[below left=.5 of v0] (i1);
\coordinate[above left=.5 of v0] (i2);
\coordinate[below right=.5 of v1] (o1);
\coordinate[above right=.5 of v1] (o2);
\draw (v0) -- (i1);
\draw (v0) -- (i2);
\draw (v1) -- (o1);
\draw (v1) -- (o2);
\draw (v0) to[bend left=20] (v2);
\draw (v0) to[bend right=20] (v3);
\draw (v1) to[bend left=20] (v3);
\draw (v1) to[bend right=20] (v2);
\draw (v2) to[bend right=60] (v3);
\draw (v2) to[bend left=60] (v3);
\filldraw (v0) circle(1pt);
\filldraw (v1) circle(1pt);
\filldraw (v2) circle(1pt);
\filldraw (v3) circle(1pt);
\ifdefined\cvl
\draw[line width=1.5pt] (v0) to[bend left=20] (v2);
\draw[line width=1.5pt] (v0) to[bend right=20] (v3);
\fi
\ifdefined\cvr
\draw[line width=1.5pt] (v1) to[bend left=20] (v3);
\draw[line width=1.5pt] (v1) to[bend right=20] (v2);
\fi
\ifdefined\cvml
\draw[line width=1.5pt] (v2) to[bend left=60] (v3);
\fi
\ifdefined\cvmr
\draw[line width=1.5pt] (v2) to[bend right=60] (v3);
\fi
\end{scope}
\end{tikzpicture}
} \right) = 
{\def \scale {4ex}
\begin{tikzpicture}[x=\scale,y=\scale,baseline={([yshift=-.5ex]current bounding box.center)}]
\begin{scope}[node distance=1]
\coordinate (top) ;
\coordinate [below left= of top] (v1);
\coordinate [below right= of top] (v2);
\coordinate [below left= of v2] (v3);
\coordinate [below= of v3] (bottom);
\draw (top) -- (v1);
\draw (top) -- (v2);
\draw (v1) -- (v3);
\draw (v2) -- (v3);
\draw (v3) -- (bottom);
\filldraw[fill=white, draw=black,circle] (top) node[fill,circle,draw,inner sep=1pt]{$3$};
\filldraw[fill=white, draw=black,circle] (v1) node[fill,circle,draw,inner sep=1pt]{$2$};
\filldraw[fill=white, draw=black,circle] (v2) node[fill,circle,draw,inner sep=1pt]{$2$};
\filldraw[fill=white, draw=black,circle] (v3) node[fill,circle,draw,inner sep=1pt]{$1$};
\filldraw[fill=white, draw=black,circle] (bottom) node[fill,circle,draw,inner sep=1pt]{$0$};
\end{scope}
\end{tikzpicture}
},
\end{align}
\else

MISSING IN DRAFT MODE

\fi
\ifdefined\nodraft
of which the reduced coproduct in $\hopfpos$ can be calculated,
\begin{align}
\label{eqn:explcopposet2}
\widetilde{\Delta}_{\hopfpos}
{\def \scale {4ex}
\begin{tikzpicture}[x=\scale,y=\scale,baseline={([yshift=-.5ex]current bounding box.center)}]
\begin{scope}[node distance=1]
\coordinate (top) ;
\coordinate [below left= of top] (v1);
\coordinate [below right= of top] (v2);
\coordinate [below left= of v2] (v3);
\coordinate [below= of v3] (bottom);
\draw (top) -- (v1);
\draw (top) -- (v2);
\draw (v1) -- (v3);
\draw (v2) -- (v3);
\draw (v3) -- (bottom);
\filldraw[fill=white, draw=black,circle] (top) node[fill,circle,draw,inner sep=1pt]{$3$};
\filldraw[fill=white, draw=black,circle] (v1) node[fill,circle,draw,inner sep=1pt]{$2$};
\filldraw[fill=white, draw=black,circle] (v2) node[fill,circle,draw,inner sep=1pt]{$2$};
\filldraw[fill=white, draw=black,circle] (v3) node[fill,circle,draw,inner sep=1pt]{$1$};
\filldraw[fill=white, draw=black,circle] (bottom) node[fill,circle,draw,inner sep=1pt]{$0$};
\end{scope}
\end{tikzpicture}
} = 
2~
{\def \scale {4ex}
\begin{tikzpicture}[x=\scale,y=\scale,baseline={([yshift=-.5ex]current bounding box.center)}]
\begin{scope}[node distance=1]
\coordinate (top) ;
\coordinate [below= of top] (v1);
\coordinate [below= of v1] (bottom);
\draw (top) -- (v1);
\draw (v1) -- (bottom);
\filldraw[fill=white, draw=black,circle] (top) node[fill,circle,draw,inner sep=1pt]{$2$};
\filldraw[fill=white, draw=black,circle] (v1) node[fill,circle,draw,inner sep=1pt]{$1$};
\filldraw[fill=white, draw=black,circle] (bottom) node[fill,circle,draw,inner sep=1pt]{$0$};
\end{scope}
\end{tikzpicture}
}
\otimes
{
\def \scale {4ex}
\begin{tikzpicture}[x=\scale,y=\scale,baseline={([yshift=-.5ex]current bounding box.center)}]
\begin{scope}[node distance=1]
\coordinate (top) ;
\coordinate [below= of top] (bottom);
\draw (top) -- (bottom);
\filldraw[fill=white, draw=black,circle] (top) node[fill,circle,draw,inner sep=1pt]{$1$};
\filldraw[fill=white, draw=black,circle] (bottom) node[fill,circle,draw,inner sep=1pt]{$0$};
\end{scope}
\end{tikzpicture}
}
+
{
\def \scale {4ex}
\begin{tikzpicture}[x=\scale,y=\scale,baseline={([yshift=-.5ex]current bounding box.center)}]
\begin{scope}[node distance=1]
\coordinate (top) ;
\coordinate [below= of top] (bottom);
\draw (top) -- (bottom);
\filldraw[fill=white, draw=black,circle] (top) node[fill,circle,draw,inner sep=1pt]{$1$};
\filldraw[fill=white, draw=black,circle] (bottom) node[fill,circle,draw,inner sep=1pt]{$0$};
\end{scope}
\end{tikzpicture}
}
\otimes
{
\def \scale {4ex}
\begin{tikzpicture}[x=\scale,y=\scale,baseline={([yshift=-.5ex]current bounding box.center)}]
\begin{scope}[node distance=1]
\coordinate (top) ;
\coordinate [below left= of top] (v1);
\coordinate [below right= of top] (v2);
\coordinate [below left= of v2] (bottom);
\draw (top) -- (v1);
\draw (top) -- (v2);
\draw (v1) -- (bottom);
\draw (v2) -- (bottom);
\filldraw[fill=white, draw=black,circle] (top) node[fill,circle,draw,inner sep=1pt]{$2$};
\filldraw[fill=white, draw=black,circle] (v1) node[fill,circle,draw,inner sep=1pt]{$1$};
\filldraw[fill=white, draw=black,circle] (v2) node[fill,circle,draw,inner sep=1pt]{$1$};
\filldraw[fill=white, draw=black,circle] (bottom) node[fill,circle,draw,inner sep=1pt]{$0$};
\end{scope}
\end{tikzpicture}
}.
\end{align}
This can be compared to the coproduct calculation, 
\begin{align}
\widetilde \Delta_4
{ 
\def \scale {2ex}
\begin{tikzpicture}[x=\scale,y=\scale,baseline={([yshift=-.5ex]current bounding box.center)}]
\begin{scope}[node distance=1]
\coordinate (v0);
\coordinate[right=.5 of v0] (v4);
\coordinate[above right= of v4] (v2);
\coordinate[below right= of v4] (v3);
\coordinate[below right= of v2] (v5);
\coordinate[right=.5 of v5] (v1);
\coordinate[above right= of v2] (o1);
\coordinate[below right= of v2] (o2);
\coordinate[below left=.5 of v0] (i1);
\coordinate[above left=.5 of v0] (i2);
\coordinate[below right=.5 of v1] (o1);
\coordinate[above right=.5 of v1] (o2);
\draw (v0) -- (i1);
\draw (v0) -- (i2);
\draw (v1) -- (o1);
\draw (v1) -- (o2);
\draw (v0) to[bend left=20] (v2);
\draw (v0) to[bend right=20] (v3);
\draw (v1) to[bend left=20] (v3);
\draw (v1) to[bend right=20] (v2);
\draw (v2) to[bend right=60] (v3);
\draw (v2) to[bend left=60] (v3);
\filldraw (v0) circle(1pt);
\filldraw (v1) circle(1pt);
\filldraw (v2) circle(1pt);
\filldraw (v3) circle(1pt);
\ifdefined\cvl
\draw[line width=1.5pt] (v0) to[bend left=20] (v2);
\draw[line width=1.5pt] (v0) to[bend right=20] (v3);
\fi
\ifdefined\cvr
\draw[line width=1.5pt] (v1) to[bend left=20] (v3);
\draw[line width=1.5pt] (v1) to[bend right=20] (v2);
\fi
\ifdefined\cvml
\draw[line width=1.5pt] (v2) to[bend left=60] (v3);
\fi
\ifdefined\cvmr
\draw[line width=1.5pt] (v2) to[bend right=60] (v3);
\fi
\end{scope}
\end{tikzpicture}
}
&= 
2 
{
\def \scale {2ex}
\begin{tikzpicture}[x=\scale,y=\scale,baseline={([yshift=-.5ex]current bounding box.center)}]
\begin{scope}[node distance=1]
\coordinate (v0);
\coordinate[right=.5 of v0] (v4);
\coordinate[above right= of v4] (v2);
\coordinate[below right= of v4] (v3);
\coordinate[above right=.5 of v2] (o1);
\coordinate[below right=.5 of v3] (o2);
\coordinate[below left=.5 of v0] (i1);
\coordinate[above left=.5 of v0] (i2);
\draw (v0) -- (i1);
\draw (v0) -- (i2);
\draw (v2) -- (o1);
\draw (v3) -- (o2);
\draw (v0) to[bend left=20] (v2);
\draw (v0) to[bend right=20] (v3);
\draw (v2) to[bend right=60] (v3);
\draw (v2) to[bend left=60] (v3);
\filldraw (v0) circle(1pt);
\filldraw (v2) circle(1pt);
\filldraw (v3) circle(1pt);
\end{scope}
\end{tikzpicture}
}
\otimes
{
\def \scale {2ex}
\begin{tikzpicture}[x=\scale,y=\scale,baseline={([yshift=-.5ex]current bounding box.center)}]
\begin{scope}[node distance=1]
    \coordinate (v0);
    \coordinate [right=.5 of v0] (vm);
    \coordinate [right=.5 of vm] (v1);
    \coordinate [above left=.5 of v0] (i0); 
    \coordinate [below left=.5 of v0] (i1); 
    \coordinate [above right=.5 of v1] (o0); 
    \coordinate [below right=.5 of v1] (o1); 
    \draw (vm) circle(.5);
    \draw (i0) -- (v0);
    \draw (i1) -- (v0);
    \draw (o0) -- (v1);
    \draw (o1) -- (v1);
    \filldraw (v0) circle(1pt);
    \filldraw (v1) circle(1pt);
\end{scope}
\end{tikzpicture}
}
+
{
\def \scale {2ex}
\begin{tikzpicture}[x=\scale,y=\scale,baseline={([yshift=-.5ex]current bounding box.center)}]
\begin{scope}[node distance=1]
    \coordinate (v0);
    \coordinate [below=1 of v0] (v1);
    \coordinate [above left=.5 of v0] (i0); 
    \coordinate [above right=.5 of v0] (i1); 
    \coordinate [below left=.5 of v1] (o0); 
    \coordinate [below right=.5 of v1] (o1); 
    \coordinate [above=.5 of v1] (vm);
    \draw (vm) circle(.5);
    \draw (i0) -- (v0);
    \draw (i1) -- (v0);
    \draw (o0) -- (v1);
    \draw (o1) -- (v1);
    \filldraw (v0) circle(1pt);
    \filldraw (v1) circle(1pt);
\end{scope}
\end{tikzpicture}
}
\otimes
{
\def \scale {2ex}
\begin{tikzpicture}[x=\scale,y=\scale,baseline={([yshift=-.5ex]current bounding box.center)}]
\begin{scope}[node distance=1]
    \coordinate (v0);
    \coordinate [right=.5 of v0] (vm1);
    \coordinate [right=.5 of vm1] (v1);
    \coordinate [right=.5 of v1] (vm2);
    \coordinate [right=.5 of vm2] (v2);
    \coordinate [above left=.5 of v0] (i0); 
    \coordinate [below left=.5 of v0] (i1); 
    \coordinate [above right=.5 of v2] (o0); 
    \coordinate [below right=.5 of v2] (o1); 
    \draw (vm1) circle(.5);
    \draw (vm2) circle(.5);
    \draw (i0) -- (v0);
    \draw (i1) -- (v0);
    \draw (o0) -- (v2);
    \draw (o1) -- (v2);
    \filldraw (v0) circle(1pt);
    \filldraw (v1) circle(1pt);
    \filldraw (v2) circle(1pt);
\end{scope}
\end{tikzpicture}
}
\end{align}
and the fact that $\chi_D$ is a Hopf algebra morphism is verified after computing the decorated poset of each subdiagram of 
$
{\def \scale {1.5ex}
\begin{tikzpicture}[x=\scale,y=\scale,baseline={([yshift=-.5ex]current bounding box.center)}]
\begin{scope}[node distance=1]
\coordinate (v0);
\coordinate[right=.5 of v0] (v4);
\coordinate[above right= of v4] (v2);
\coordinate[below right= of v4] (v3);
\coordinate[below right= of v2] (v5);
\coordinate[right=.5 of v5] (v1);
\coordinate[above right= of v2] (o1);
\coordinate[below right= of v2] (o2);
\coordinate[below left=.5 of v0] (i1);
\coordinate[above left=.5 of v0] (i2);
\coordinate[below right=.5 of v1] (o1);
\coordinate[above right=.5 of v1] (o2);
\draw (v0) -- (i1);
\draw (v0) -- (i2);
\draw (v1) -- (o1);
\draw (v1) -- (o2);
\draw (v0) to[bend left=20] (v2);
\draw (v0) to[bend right=20] (v3);
\draw (v1) to[bend left=20] (v3);
\draw (v1) to[bend right=20] (v2);
\draw (v2) to[bend right=60] (v3);
\draw (v2) to[bend left=60] (v3);
\filldraw (v0) circle(1pt);
\filldraw (v1) circle(1pt);
\filldraw (v2) circle(1pt);
\filldraw (v3) circle(1pt);
\ifdefined\cvl
\draw[line width=1.5pt] (v0) to[bend left=20] (v2);
\draw[line width=1.5pt] (v0) to[bend right=20] (v3);
\fi
\ifdefined\cvr
\draw[line width=1.5pt] (v1) to[bend left=20] (v3);
\draw[line width=1.5pt] (v1) to[bend right=20] (v2);
\fi
\ifdefined\cvml
\draw[line width=1.5pt] (v2) to[bend left=60] (v3);
\fi
\ifdefined\cvmr
\draw[line width=1.5pt] (v2) to[bend right=60] (v3);
\fi
\end{scope}
\end{tikzpicture}
}$ and comparing the previous two equations:
\begin{align}
&\chi_4 \left(
{\def \scale {4ex}
\begin{tikzpicture}[x=\scale,y=\scale,baseline={([yshift=-.5ex]current bounding box.center)}]
\begin{scope}[node distance=1]
\coordinate (v0);
\coordinate[right=.5 of v0] (v4);
\coordinate[above right= of v4] (v2);
\coordinate[below right= of v4] (v3);
\coordinate[above right=.5 of v2] (o1);
\coordinate[below right=.5 of v3] (o2);
\coordinate[below left=.5 of v0] (i1);
\coordinate[above left=.5 of v0] (i2);
\draw (v0) -- (i1);
\draw (v0) -- (i2);
\draw (v2) -- (o1);
\draw (v3) -- (o2);
\draw (v0) to[bend left=20] (v2);
\draw (v0) to[bend right=20] (v3);
\draw (v2) to[bend right=60] (v3);
\draw (v2) to[bend left=60] (v3);
\filldraw (v0) circle(1pt);
\filldraw (v2) circle(1pt);
\filldraw (v3) circle(1pt);
\end{scope}
\end{tikzpicture}
}
\right) = {\def \scale {4ex}
\begin{tikzpicture}[x=\scale,y=\scale,baseline={([yshift=-.5ex]current bounding box.center)}]
\begin{scope}[node distance=1]
\coordinate (top) ;
\coordinate [below= of top] (v1);
\coordinate [below= of v1] (bottom);
\draw (top) -- (v1);
\draw (v1) -- (bottom);
\filldraw[fill=white, draw=black,circle] (top) node[fill,circle,draw,inner sep=1pt]{$2$};
\filldraw[fill=white, draw=black,circle] (v1) node[fill,circle,draw,inner sep=1pt]{$1$};
\filldraw[fill=white, draw=black,circle] (bottom) node[fill,circle,draw,inner sep=1pt]{$0$};
\end{scope}
\end{tikzpicture}
}& 
&\chi_4 \left(
{
\def \scale {4ex}
\begin{tikzpicture}[x=\scale,y=\scale,baseline={([yshift=-.5ex]current bounding box.center)}]
\begin{scope}[node distance=1]
    \coordinate (v0);
    \coordinate [right=.5 of v0] (vm);
    \coordinate [right=.5 of vm] (v1);
    \coordinate [above left=.5 of v0] (i0); 
    \coordinate [below left=.5 of v0] (i1); 
    \coordinate [above right=.5 of v1] (o0); 
    \coordinate [below right=.5 of v1] (o1); 
    \draw (vm) circle(.5);
    \draw (i0) -- (v0);
    \draw (i1) -- (v0);
    \draw (o0) -- (v1);
    \draw (o1) -- (v1);
    \filldraw (v0) circle(1pt);
    \filldraw (v1) circle(1pt);
\end{scope}
\end{tikzpicture}
} \right)
=
{
\def \scale {4ex}
\begin{tikzpicture}[x=\scale,y=\scale,baseline={([yshift=-.5ex]current bounding box.center)}]
\begin{scope}[node distance=1]
\coordinate (top) ;
\coordinate [below= of top] (bottom);
\draw (top) -- (bottom);
\filldraw[fill=white, draw=black,circle] (top) node[fill,circle,draw,inner sep=1pt]{$1$};
\filldraw[fill=white, draw=black,circle] (bottom) node[fill,circle,draw,inner sep=1pt]{$0$};
\end{scope}
\end{tikzpicture}
}
&
&\chi_4 \left(
{\def \scale {4ex}
\begin{tikzpicture}[x=\scale,y=\scale,baseline={([yshift=-.5ex]current bounding box.center)}]
\begin{scope}[node distance=1]
    \coordinate (v0);
    \coordinate [right=.5 of v0] (vm1);
    \coordinate [right=.5 of vm1] (v1);
    \coordinate [right=.5 of v1] (vm2);
    \coordinate [right=.5 of vm2] (v2);
    \coordinate [above left=.5 of v0] (i0); 
    \coordinate [below left=.5 of v0] (i1); 
    \coordinate [above right=.5 of v2] (o0); 
    \coordinate [below right=.5 of v2] (o1); 
    \draw (vm1) circle(.5);
    \draw (vm2) circle(.5);
    \draw (i0) -- (v0);
    \draw (i1) -- (v0);
    \draw (o0) -- (v2);
    \draw (o1) -- (v2);
    \filldraw (v0) circle(1pt);
    \filldraw (v1) circle(1pt);
    \filldraw (v2) circle(1pt);
\end{scope}
\end{tikzpicture}
}
\right) = {\def \scale {4ex}
\begin{tikzpicture}[x=\scale,y=\scale,baseline={([yshift=-.5ex]current bounding box.center)}]
\begin{scope}[node distance=1]
\coordinate (top) ;
\coordinate [below left= of top] (v1);
\coordinate [below right= of top] (v2);
\coordinate [below left= of v2] (bottom);
\draw (top) -- (v1);
\draw (top) -- (v2);
\draw (v1) -- (bottom);
\draw (v2) -- (bottom);
\filldraw[fill=white, draw=black,circle] (top) node[fill,circle,draw,inner sep=1pt]{$2$};
\filldraw[fill=white, draw=black,circle] (v1) node[fill,circle,draw,inner sep=1pt]{$1$};
\filldraw[fill=white, draw=black,circle] (v2) node[fill,circle,draw,inner sep=1pt]{$1$};
\filldraw[fill=white, draw=black,circle] (bottom) node[fill,circle,draw,inner sep=1pt]{$0$};
\end{scope}
\end{tikzpicture}
}.
\end{align}
\else

MISSING IN DRAFT MODE

\fi
\end{expl}

\section{An application of the Hopf algebra of decorated lattices}
Some calculations are easily performed in the Hopf algebra of decorated lattices, but hard do on the Feynman diagram counterpart. One example is the evaluation of the counterterm in zero-dimensional QFTs, where the Feynman rules map every diagram to a constant. The counterterm map in zero-dimensional field theory takes the form 
\begin{align}
\label{eqn:antipode}
 S^R_D := \phi \circ S_D,
\end{align}
where $\phi$ are the Feynman rules, which map $\Gamma$ to $\hbar^{h_1(\Gamma)}$ and $S_D$ is the antipode on $\hopffg_D$.

Using the fact that $\chi_D$ is a Hopf algebra morphism it can be shown that
\begin{prop}{\cite[Corr. 5]{borinsky2015}}
\begin{align}
\label{eqn:phimoebius}
{S^R_D}(\Gamma) = \hbar^{h_1(\Gamma)} \mu_{\sdsubdiags_D(\Gamma)}( \hat{0}, \hat{1} )
\end{align}
on the Hopf algebra of Feynman diagrams with $\hat{0}=\emptyset$ and $\hat{1}=\Gamma$, the lower and upper bound of $\sdsubdiags_D(\Gamma)$, where $S^R_D$ is the counterterm map in zero-dimensional field theory and $\mu_L$ the Moebius function of the lattice $L$. 
The Moebius function is defined as,
\begin{align}
\label{eqn:moebius}
\mu_P(x,y) &= 
\begin{cases}
  1,&\text{if } x=y  \\
- \sum \limits_{x \leq z < y} \mu_P(x,z)  & \text{if } x < y.
\end{cases}
\end{align}
for a poset $P$ and $x,y \in P$.
\end{prop}
The calculation of the Moebius function is in general much easier than the calculation of the antipode in formula \eqref{eqn:antipode}. This statement can also be used to deduce generating functions for the weighted number of primitive diagrams in QFTs as was done for $\phi^4$ and Yang-Mills in terms of the counter terms in \cite{borinsky2015}. In a future publication, these ideas will be used to enumerate the weighted number of primitive diagrams for these theories explicitly \cite{borinskyunpub}. 

\section{Properties of the lattices}
Having established a connection between the Hopf algebra of Feynman diagrams and the lattices, we can ask what the lattices tell us about the coradical degree of the diagrams. It is easily seen from the definition of the coproducts in $\hopflat$ and $\hopffg$ that the length of the longest `chain', a path from top of the Hasse diagram to the bottom, is the coradical degree of the Feynman diagram. If all complete chains have the same length, this number is called the rank of the poset or lattice and the poset or lattice is called ranked or graded. 

A chain of the poset $\sdsubdiags_D(\Gamma)$ corresponds to a \textit{forest} of the diagram in the scope of the BPHZ algorithm. The statement that the poset $\sdsubdiags_D(\Gamma)$ is graded implies that all complete forests of the diagram have the same cardinality. Furthermore, it means that the coradical filtration is in fact a graduation of the Hopf algebra of Feynman diagrams \cite{borinsky2015}. 

{
\ifdefined\nodraft
\begin{figure}
  \subcaptionbox{Example of a diagram where $\sdsubdiags_4(\Gamma)$ forms a non-graded lattice.\label{fig:diagramnonranked}}
  [.45\linewidth]{
      $\Gamma = {
\def \scale{2em}
\begin{tikzpicture}[x=\scale,y=\scale,baseline={([yshift=-.5ex]current bounding box.center)}]
\begin{scope}[node distance=1]
\coordinate (i1);
\coordinate[right=.5 of i1] (v0);
\coordinate[right=.5 of v0] (v2);
\coordinate[right=.5 of v2] (vm);
\coordinate[above=of vm] (v3);
\coordinate[below=of vm] (v4);
\coordinate[right=.5 of vm] (v2d);
\coordinate[right=.5 of v2d] (v1);
\coordinate[right=.5 of v1] (o1);
\ifdefined \cvlt
\draw[line width=1.5pt] (v0) to[bend left=45] (v3);
\else
\draw (v0) to[bend left=45] (v3);
\fi
\ifdefined \cvlb
\draw[line width=1.5pt] (v0) to[bend right=45] (v4);
\else
\draw (v0) to[bend right=45] (v4);
\fi
\ifdefined \cvlm
\draw[line width=1.5pt] (v0) -- (v2);
\else
\draw (v0) -- (v2);
\fi
\ifdefined \cvrt
\draw[line width=1.5pt] (v1) to[bend right=45] (v3);
\else
\draw (v1) to[bend right=45] (v3);
\fi
\ifdefined \cvrb
\draw[line width=1.5pt] (v1) to[bend left=45] (v4);
\else
\draw (v1) to[bend left=45] (v4);
\fi
\ifdefined \cvmt
\draw[line width=1.5pt] (v3) to[bend right=20] (v2);
\else
\draw (v3) to[bend right=20] (v2);
\fi
\ifdefined \cvmb
\draw[line width=1.5pt] (v4) to[bend left=20] (v2);
\else
\draw (v4) to[bend left=20] (v2);
\fi
\ifdefined \cvmm
\draw[line width=1.5pt] (v3) to[bend left=20] (v2d);
\draw[line width=1.5pt] (v4) to[bend right=20] (v2d);
\else
\draw (v3) to[bend left=20] (v2d);
\draw (v4) to[bend right=20] (v2d);
\fi
\filldraw[color=white] (v2d) circle(.2);
\ifdefined \cvrm
\draw[line width=1.5pt] (v1) -- (v2);
\else
\draw (v1) -- (v2);
\fi
\draw (v0) -- (i1);
\draw (v1) -- (o1);
\filldraw (v0) circle(1pt);
\filldraw (v1) circle(1pt);
\filldraw (v2) circle(1pt);
\filldraw (v3) circle(1pt);
\filldraw (v4) circle(1pt);
\end{scope}
\end{tikzpicture}
      }$
  }
  \subcaptionbox{The Hasse diagram of the corresponding non-graded lattice, where the decoration was omitted.\label{fig:diagramnonrankedlattice}}
  [.45\linewidth]{
  $\chi_4(\Gamma) = 
{
\def \scale{1em}
\begin{tikzpicture}[x=\scale,y=\scale,baseline={([yshift=-.5ex]current bounding box.center)}]
\begin{scope}[node distance=1]
\coordinate (top) ;
\coordinate [below left=2 and   .5 of top] (v1) ;
\coordinate [below left=2 and  1   of top] (v2) ;
\coordinate [below left=2 and  1.5 of top] (v3) ;
\coordinate [below right=1.5  and  1   of top] (w1) ;
\coordinate [below right=1.5  and  1.5 of top] (w2) ;
\coordinate [below right=1.5  and  2   of top] (w3) ;
\coordinate [below right=1.5  and  3   of top] (w4) ;
\coordinate [below right=1.5  and  3.5 of top] (w5) ;
\coordinate [below right=1.5  and  4   of top] (w6) ;
\coordinate [below right=2.5  and  1.5 of top] (u1) ;
\coordinate [below right=2.5  and  3.5 of top] (u2) ;
\coordinate [below=4 of top] (bottom) ;
\draw (top) -- (v1);
\draw (top) -- (v2);
\draw (top) -- (v3);
\draw (top) to[out=-30,in=90] (w1);
\draw (top) to[out=-25,in=90] (w2);
\draw (top) to[out=-20,in=90] (w3);
\draw (top) to[out=-15,in=90] (w4);
\draw (top) to[out=-10,in=90] (w5);
\draw (top) to[out=-5,in=90] (w6);
\draw (w1) -- (u1);
\draw (w2) -- (u1);
\draw (w3) -- (u1);
\draw (w4) -- (u2);
\draw (w5) -- (u2);
\draw (w6) -- (u2);
\draw (v1) -- (bottom);
\draw (v2) -- (bottom);
\draw (v3) -- (bottom);
\draw (u1) -- (bottom);
\draw (u2) -- (bottom);
\filldraw[fill=white, draw=black] (top) circle(2pt);
\filldraw[fill=white, draw=black] (v1) circle(2pt);
\filldraw[fill=white, draw=black] (v2) circle(2pt);
\filldraw[fill=white, draw=black] (v3) circle(2pt);
\filldraw[fill=white, draw=black] (w1) circle(2pt);
\filldraw[fill=white, draw=black] (w2) circle(2pt);
\filldraw[fill=white, draw=black] (w3) circle(2pt);
\filldraw[fill=white, draw=black] (w4) circle(2pt);
\filldraw[fill=white, draw=black] (w5) circle(2pt);
\filldraw[fill=white, draw=black] (w6) circle(2pt);
\filldraw[fill=white, draw=black] (u1) circle(2pt);
\filldraw[fill=white, draw=black] (u2) circle(2pt);
\filldraw[fill=white, draw=black] (bottom) circle(2pt);
\end{scope}
\end{tikzpicture}
}$
  }
  \subcaptionbox{The non-trivial superficially divergent subdiagrams and the complete forests which can be formed out of them.\label{fig:diagramnonrankeddetails}}
  [\linewidth]{
\def \scale{1.5em}
$
\begin{aligned}
\alpha_1 &= 
{
\def\cvlm {}
\def\cvlb {}
\def\cvmt {}
\def\cvmm {}
\def\cvmb {}
\def\cvrt {}
\def\cvrm {}
\def\cvrb {}
\begin{tikzpicture}[x=\scale,y=\scale,baseline={([yshift=-.5ex]current bounding box.center)}]
\begin{scope}[node distance=1]
\coordinate (i1);
\coordinate[right=.5 of i1] (v0);
\coordinate[right=.5 of v0] (v2);
\coordinate[right=.5 of v2] (vm);
\coordinate[above=of vm] (v3);
\coordinate[below=of vm] (v4);
\coordinate[right=.5 of vm] (v2d);
\coordinate[right=.5 of v2d] (v1);
\coordinate[right=.5 of v1] (o1);
\ifdefined \cvlt
\draw[line width=1.5pt] (v0) to[bend left=45] (v3);
\else
\draw (v0) to[bend left=45] (v3);
\fi
\ifdefined \cvlb
\draw[line width=1.5pt] (v0) to[bend right=45] (v4);
\else
\draw (v0) to[bend right=45] (v4);
\fi
\ifdefined \cvlm
\draw[line width=1.5pt] (v0) -- (v2);
\else
\draw (v0) -- (v2);
\fi
\ifdefined \cvrt
\draw[line width=1.5pt] (v1) to[bend right=45] (v3);
\else
\draw (v1) to[bend right=45] (v3);
\fi
\ifdefined \cvrb
\draw[line width=1.5pt] (v1) to[bend left=45] (v4);
\else
\draw (v1) to[bend left=45] (v4);
\fi
\ifdefined \cvmt
\draw[line width=1.5pt] (v3) to[bend right=20] (v2);
\else
\draw (v3) to[bend right=20] (v2);
\fi
\ifdefined \cvmb
\draw[line width=1.5pt] (v4) to[bend left=20] (v2);
\else
\draw (v4) to[bend left=20] (v2);
\fi
\ifdefined \cvmm
\draw[line width=1.5pt] (v3) to[bend left=20] (v2d);
\draw[line width=1.5pt] (v4) to[bend right=20] (v2d);
\else
\draw (v3) to[bend left=20] (v2d);
\draw (v4) to[bend right=20] (v2d);
\fi
\filldraw[color=white] (v2d) circle(.2);
\ifdefined \cvrm
\draw[line width=1.5pt] (v1) -- (v2);
\else
\draw (v1) -- (v2);
\fi
\draw (v0) -- (i1);
\draw (v1) -- (o1);
\filldraw (v0) circle(1pt);
\filldraw (v1) circle(1pt);
\filldraw (v2) circle(1pt);
\filldraw (v3) circle(1pt);
\filldraw (v4) circle(1pt);
\end{scope}
\end{tikzpicture}
}
, &
\alpha_2 &= 
{
\def\cvlt {}
\def\cvlb {}
\def\cvmt {}
\def\cvmm {}
\def\cvmb {}
\def\cvrt {}
\def\cvrm {}
\def\cvrb {}
\begin{tikzpicture}[x=\scale,y=\scale,baseline={([yshift=-.5ex]current bounding box.center)}]
\begin{scope}[node distance=1]
\coordinate (i1);
\coordinate[right=.5 of i1] (v0);
\coordinate[right=.5 of v0] (v2);
\coordinate[right=.5 of v2] (vm);
\coordinate[above=of vm] (v3);
\coordinate[below=of vm] (v4);
\coordinate[right=.5 of vm] (v2d);
\coordinate[right=.5 of v2d] (v1);
\coordinate[right=.5 of v1] (o1);
\ifdefined \cvlt
\draw[line width=1.5pt] (v0) to[bend left=45] (v3);
\else
\draw (v0) to[bend left=45] (v3);
\fi
\ifdefined \cvlb
\draw[line width=1.5pt] (v0) to[bend right=45] (v4);
\else
\draw (v0) to[bend right=45] (v4);
\fi
\ifdefined \cvlm
\draw[line width=1.5pt] (v0) -- (v2);
\else
\draw (v0) -- (v2);
\fi
\ifdefined \cvrt
\draw[line width=1.5pt] (v1) to[bend right=45] (v3);
\else
\draw (v1) to[bend right=45] (v3);
\fi
\ifdefined \cvrb
\draw[line width=1.5pt] (v1) to[bend left=45] (v4);
\else
\draw (v1) to[bend left=45] (v4);
\fi
\ifdefined \cvmt
\draw[line width=1.5pt] (v3) to[bend right=20] (v2);
\else
\draw (v3) to[bend right=20] (v2);
\fi
\ifdefined \cvmb
\draw[line width=1.5pt] (v4) to[bend left=20] (v2);
\else
\draw (v4) to[bend left=20] (v2);
\fi
\ifdefined \cvmm
\draw[line width=1.5pt] (v3) to[bend left=20] (v2d);
\draw[line width=1.5pt] (v4) to[bend right=20] (v2d);
\else
\draw (v3) to[bend left=20] (v2d);
\draw (v4) to[bend right=20] (v2d);
\fi
\filldraw[color=white] (v2d) circle(.2);
\ifdefined \cvrm
\draw[line width=1.5pt] (v1) -- (v2);
\else
\draw (v1) -- (v2);
\fi
\draw (v0) -- (i1);
\draw (v1) -- (o1);
\filldraw (v0) circle(1pt);
\filldraw (v1) circle(1pt);
\filldraw (v2) circle(1pt);
\filldraw (v3) circle(1pt);
\filldraw (v4) circle(1pt);
\end{scope}
\end{tikzpicture}
}
, &
\alpha_3 &= 
{
\def\cvlm {}
\def\cvlt {}
\def\cvmt {}
\def\cvmm {}
\def\cvmb {}
\def\cvrt {}
\def\cvrm {}
\def\cvrb {}
\begin{tikzpicture}[x=\scale,y=\scale,baseline={([yshift=-.5ex]current bounding box.center)}]
\begin{scope}[node distance=1]
\coordinate (i1);
\coordinate[right=.5 of i1] (v0);
\coordinate[right=.5 of v0] (v2);
\coordinate[right=.5 of v2] (vm);
\coordinate[above=of vm] (v3);
\coordinate[below=of vm] (v4);
\coordinate[right=.5 of vm] (v2d);
\coordinate[right=.5 of v2d] (v1);
\coordinate[right=.5 of v1] (o1);
\ifdefined \cvlt
\draw[line width=1.5pt] (v0) to[bend left=45] (v3);
\else
\draw (v0) to[bend left=45] (v3);
\fi
\ifdefined \cvlb
\draw[line width=1.5pt] (v0) to[bend right=45] (v4);
\else
\draw (v0) to[bend right=45] (v4);
\fi
\ifdefined \cvlm
\draw[line width=1.5pt] (v0) -- (v2);
\else
\draw (v0) -- (v2);
\fi
\ifdefined \cvrt
\draw[line width=1.5pt] (v1) to[bend right=45] (v3);
\else
\draw (v1) to[bend right=45] (v3);
\fi
\ifdefined \cvrb
\draw[line width=1.5pt] (v1) to[bend left=45] (v4);
\else
\draw (v1) to[bend left=45] (v4);
\fi
\ifdefined \cvmt
\draw[line width=1.5pt] (v3) to[bend right=20] (v2);
\else
\draw (v3) to[bend right=20] (v2);
\fi
\ifdefined \cvmb
\draw[line width=1.5pt] (v4) to[bend left=20] (v2);
\else
\draw (v4) to[bend left=20] (v2);
\fi
\ifdefined \cvmm
\draw[line width=1.5pt] (v3) to[bend left=20] (v2d);
\draw[line width=1.5pt] (v4) to[bend right=20] (v2d);
\else
\draw (v3) to[bend left=20] (v2d);
\draw (v4) to[bend right=20] (v2d);
\fi
\filldraw[color=white] (v2d) circle(.2);
\ifdefined \cvrm
\draw[line width=1.5pt] (v1) -- (v2);
\else
\draw (v1) -- (v2);
\fi
\draw (v0) -- (i1);
\draw (v1) -- (o1);
\filldraw (v0) circle(1pt);
\filldraw (v1) circle(1pt);
\filldraw (v2) circle(1pt);
\filldraw (v3) circle(1pt);
\filldraw (v4) circle(1pt);
\end{scope}
\end{tikzpicture}
}
\\
\beta_1 &= 
{
\def\cvrm {}
\def\cvrb {}
\def\cvmt {}
\def\cvmm {}
\def\cvmb {}
\def\cvlt {}
\def\cvlm {}
\def\cvlb {}
\begin{tikzpicture}[x=\scale,y=\scale,baseline={([yshift=-.5ex]current bounding box.center)}]
\begin{scope}[node distance=1]
\coordinate (i1);
\coordinate[right=.5 of i1] (v0);
\coordinate[right=.5 of v0] (v2);
\coordinate[right=.5 of v2] (vm);
\coordinate[above=of vm] (v3);
\coordinate[below=of vm] (v4);
\coordinate[right=.5 of vm] (v2d);
\coordinate[right=.5 of v2d] (v1);
\coordinate[right=.5 of v1] (o1);
\ifdefined \cvlt
\draw[line width=1.5pt] (v0) to[bend left=45] (v3);
\else
\draw (v0) to[bend left=45] (v3);
\fi
\ifdefined \cvlb
\draw[line width=1.5pt] (v0) to[bend right=45] (v4);
\else
\draw (v0) to[bend right=45] (v4);
\fi
\ifdefined \cvlm
\draw[line width=1.5pt] (v0) -- (v2);
\else
\draw (v0) -- (v2);
\fi
\ifdefined \cvrt
\draw[line width=1.5pt] (v1) to[bend right=45] (v3);
\else
\draw (v1) to[bend right=45] (v3);
\fi
\ifdefined \cvrb
\draw[line width=1.5pt] (v1) to[bend left=45] (v4);
\else
\draw (v1) to[bend left=45] (v4);
\fi
\ifdefined \cvmt
\draw[line width=1.5pt] (v3) to[bend right=20] (v2);
\else
\draw (v3) to[bend right=20] (v2);
\fi
\ifdefined \cvmb
\draw[line width=1.5pt] (v4) to[bend left=20] (v2);
\else
\draw (v4) to[bend left=20] (v2);
\fi
\ifdefined \cvmm
\draw[line width=1.5pt] (v3) to[bend left=20] (v2d);
\draw[line width=1.5pt] (v4) to[bend right=20] (v2d);
\else
\draw (v3) to[bend left=20] (v2d);
\draw (v4) to[bend right=20] (v2d);
\fi
\filldraw[color=white] (v2d) circle(.2);
\ifdefined \cvrm
\draw[line width=1.5pt] (v1) -- (v2);
\else
\draw (v1) -- (v2);
\fi
\draw (v0) -- (i1);
\draw (v1) -- (o1);
\filldraw (v0) circle(1pt);
\filldraw (v1) circle(1pt);
\filldraw (v2) circle(1pt);
\filldraw (v3) circle(1pt);
\filldraw (v4) circle(1pt);
\end{scope}
\end{tikzpicture}
}
, &
\beta_2 &= 
{
\def\cvrt {}
\def\cvrb {}
\def\cvmt {}
\def\cvmm {}
\def\cvmb {}
\def\cvlt {}
\def\cvlm {}
\def\cvlb {}
\begin{tikzpicture}[x=\scale,y=\scale,baseline={([yshift=-.5ex]current bounding box.center)}]
\begin{scope}[node distance=1]
\coordinate (i1);
\coordinate[right=.5 of i1] (v0);
\coordinate[right=.5 of v0] (v2);
\coordinate[right=.5 of v2] (vm);
\coordinate[above=of vm] (v3);
\coordinate[below=of vm] (v4);
\coordinate[right=.5 of vm] (v2d);
\coordinate[right=.5 of v2d] (v1);
\coordinate[right=.5 of v1] (o1);
\ifdefined \cvlt
\draw[line width=1.5pt] (v0) to[bend left=45] (v3);
\else
\draw (v0) to[bend left=45] (v3);
\fi
\ifdefined \cvlb
\draw[line width=1.5pt] (v0) to[bend right=45] (v4);
\else
\draw (v0) to[bend right=45] (v4);
\fi
\ifdefined \cvlm
\draw[line width=1.5pt] (v0) -- (v2);
\else
\draw (v0) -- (v2);
\fi
\ifdefined \cvrt
\draw[line width=1.5pt] (v1) to[bend right=45] (v3);
\else
\draw (v1) to[bend right=45] (v3);
\fi
\ifdefined \cvrb
\draw[line width=1.5pt] (v1) to[bend left=45] (v4);
\else
\draw (v1) to[bend left=45] (v4);
\fi
\ifdefined \cvmt
\draw[line width=1.5pt] (v3) to[bend right=20] (v2);
\else
\draw (v3) to[bend right=20] (v2);
\fi
\ifdefined \cvmb
\draw[line width=1.5pt] (v4) to[bend left=20] (v2);
\else
\draw (v4) to[bend left=20] (v2);
\fi
\ifdefined \cvmm
\draw[line width=1.5pt] (v3) to[bend left=20] (v2d);
\draw[line width=1.5pt] (v4) to[bend right=20] (v2d);
\else
\draw (v3) to[bend left=20] (v2d);
\draw (v4) to[bend right=20] (v2d);
\fi
\filldraw[color=white] (v2d) circle(.2);
\ifdefined \cvrm
\draw[line width=1.5pt] (v1) -- (v2);
\else
\draw (v1) -- (v2);
\fi
\draw (v0) -- (i1);
\draw (v1) -- (o1);
\filldraw (v0) circle(1pt);
\filldraw (v1) circle(1pt);
\filldraw (v2) circle(1pt);
\filldraw (v3) circle(1pt);
\filldraw (v4) circle(1pt);
\end{scope}
\end{tikzpicture}
}
, &
\beta_3 &= 
{
\def\cvrm {}
\def\cvrt {}
\def\cvmt {}
\def\cvmm {}
\def\cvmb {}
\def\cvlt {}
\def\cvlm {}
\def\cvlb {}
\begin{tikzpicture}[x=\scale,y=\scale,baseline={([yshift=-.5ex]current bounding box.center)}]
\begin{scope}[node distance=1]
\coordinate (i1);
\coordinate[right=.5 of i1] (v0);
\coordinate[right=.5 of v0] (v2);
\coordinate[right=.5 of v2] (vm);
\coordinate[above=of vm] (v3);
\coordinate[below=of vm] (v4);
\coordinate[right=.5 of vm] (v2d);
\coordinate[right=.5 of v2d] (v1);
\coordinate[right=.5 of v1] (o1);
\ifdefined \cvlt
\draw[line width=1.5pt] (v0) to[bend left=45] (v3);
\else
\draw (v0) to[bend left=45] (v3);
\fi
\ifdefined \cvlb
\draw[line width=1.5pt] (v0) to[bend right=45] (v4);
\else
\draw (v0) to[bend right=45] (v4);
\fi
\ifdefined \cvlm
\draw[line width=1.5pt] (v0) -- (v2);
\else
\draw (v0) -- (v2);
\fi
\ifdefined \cvrt
\draw[line width=1.5pt] (v1) to[bend right=45] (v3);
\else
\draw (v1) to[bend right=45] (v3);
\fi
\ifdefined \cvrb
\draw[line width=1.5pt] (v1) to[bend left=45] (v4);
\else
\draw (v1) to[bend left=45] (v4);
\fi
\ifdefined \cvmt
\draw[line width=1.5pt] (v3) to[bend right=20] (v2);
\else
\draw (v3) to[bend right=20] (v2);
\fi
\ifdefined \cvmb
\draw[line width=1.5pt] (v4) to[bend left=20] (v2);
\else
\draw (v4) to[bend left=20] (v2);
\fi
\ifdefined \cvmm
\draw[line width=1.5pt] (v3) to[bend left=20] (v2d);
\draw[line width=1.5pt] (v4) to[bend right=20] (v2d);
\else
\draw (v3) to[bend left=20] (v2d);
\draw (v4) to[bend right=20] (v2d);
\fi
\filldraw[color=white] (v2d) circle(.2);
\ifdefined \cvrm
\draw[line width=1.5pt] (v1) -- (v2);
\else
\draw (v1) -- (v2);
\fi
\draw (v0) -- (i1);
\draw (v1) -- (o1);
\filldraw (v0) circle(1pt);
\filldraw (v1) circle(1pt);
\filldraw (v2) circle(1pt);
\filldraw (v3) circle(1pt);
\filldraw (v4) circle(1pt);
\end{scope}
\end{tikzpicture}
}
\\
\gamma_1 &= 
{
\def\cvmm {}
\def\cvmb {}
\def\cvrt {}
\def\cvrm {}
\def\cvrb {}
\def\cvlt {}
\def\cvlm {}
\def\cvlb {}
\begin{tikzpicture}[x=\scale,y=\scale,baseline={([yshift=-.5ex]current bounding box.center)}]
\begin{scope}[node distance=1]
\coordinate (i1);
\coordinate[right=.5 of i1] (v0);
\coordinate[right=.5 of v0] (v2);
\coordinate[right=.5 of v2] (vm);
\coordinate[above=of vm] (v3);
\coordinate[below=of vm] (v4);
\coordinate[right=.5 of vm] (v2d);
\coordinate[right=.5 of v2d] (v1);
\coordinate[right=.5 of v1] (o1);
\ifdefined \cvlt
\draw[line width=1.5pt] (v0) to[bend left=45] (v3);
\else
\draw (v0) to[bend left=45] (v3);
\fi
\ifdefined \cvlb
\draw[line width=1.5pt] (v0) to[bend right=45] (v4);
\else
\draw (v0) to[bend right=45] (v4);
\fi
\ifdefined \cvlm
\draw[line width=1.5pt] (v0) -- (v2);
\else
\draw (v0) -- (v2);
\fi
\ifdefined \cvrt
\draw[line width=1.5pt] (v1) to[bend right=45] (v3);
\else
\draw (v1) to[bend right=45] (v3);
\fi
\ifdefined \cvrb
\draw[line width=1.5pt] (v1) to[bend left=45] (v4);
\else
\draw (v1) to[bend left=45] (v4);
\fi
\ifdefined \cvmt
\draw[line width=1.5pt] (v3) to[bend right=20] (v2);
\else
\draw (v3) to[bend right=20] (v2);
\fi
\ifdefined \cvmb
\draw[line width=1.5pt] (v4) to[bend left=20] (v2);
\else
\draw (v4) to[bend left=20] (v2);
\fi
\ifdefined \cvmm
\draw[line width=1.5pt] (v3) to[bend left=20] (v2d);
\draw[line width=1.5pt] (v4) to[bend right=20] (v2d);
\else
\draw (v3) to[bend left=20] (v2d);
\draw (v4) to[bend right=20] (v2d);
\fi
\filldraw[color=white] (v2d) circle(.2);
\ifdefined \cvrm
\draw[line width=1.5pt] (v1) -- (v2);
\else
\draw (v1) -- (v2);
\fi
\draw (v0) -- (i1);
\draw (v1) -- (o1);
\filldraw (v0) circle(1pt);
\filldraw (v1) circle(1pt);
\filldraw (v2) circle(1pt);
\filldraw (v3) circle(1pt);
\filldraw (v4) circle(1pt);
\end{scope}
\end{tikzpicture}
}
, &
\gamma_2 &= 
{
\def\cvmt {}
\def\cvmb {}
\def\cvrt {}
\def\cvrm {}
\def\cvrb {}
\def\cvlt {}
\def\cvlm {}
\def\cvlb {}
\begin{tikzpicture}[x=\scale,y=\scale,baseline={([yshift=-.5ex]current bounding box.center)}]
\begin{scope}[node distance=1]
\coordinate (i1);
\coordinate[right=.5 of i1] (v0);
\coordinate[right=.5 of v0] (v2);
\coordinate[right=.5 of v2] (vm);
\coordinate[above=of vm] (v3);
\coordinate[below=of vm] (v4);
\coordinate[right=.5 of vm] (v2d);
\coordinate[right=.5 of v2d] (v1);
\coordinate[right=.5 of v1] (o1);
\ifdefined \cvlt
\draw[line width=1.5pt] (v0) to[bend left=45] (v3);
\else
\draw (v0) to[bend left=45] (v3);
\fi
\ifdefined \cvlb
\draw[line width=1.5pt] (v0) to[bend right=45] (v4);
\else
\draw (v0) to[bend right=45] (v4);
\fi
\ifdefined \cvlm
\draw[line width=1.5pt] (v0) -- (v2);
\else
\draw (v0) -- (v2);
\fi
\ifdefined \cvrt
\draw[line width=1.5pt] (v1) to[bend right=45] (v3);
\else
\draw (v1) to[bend right=45] (v3);
\fi
\ifdefined \cvrb
\draw[line width=1.5pt] (v1) to[bend left=45] (v4);
\else
\draw (v1) to[bend left=45] (v4);
\fi
\ifdefined \cvmt
\draw[line width=1.5pt] (v3) to[bend right=20] (v2);
\else
\draw (v3) to[bend right=20] (v2);
\fi
\ifdefined \cvmb
\draw[line width=1.5pt] (v4) to[bend left=20] (v2);
\else
\draw (v4) to[bend left=20] (v2);
\fi
\ifdefined \cvmm
\draw[line width=1.5pt] (v3) to[bend left=20] (v2d);
\draw[line width=1.5pt] (v4) to[bend right=20] (v2d);
\else
\draw (v3) to[bend left=20] (v2d);
\draw (v4) to[bend right=20] (v2d);
\fi
\filldraw[color=white] (v2d) circle(.2);
\ifdefined \cvrm
\draw[line width=1.5pt] (v1) -- (v2);
\else
\draw (v1) -- (v2);
\fi
\draw (v0) -- (i1);
\draw (v1) -- (o1);
\filldraw (v0) circle(1pt);
\filldraw (v1) circle(1pt);
\filldraw (v2) circle(1pt);
\filldraw (v3) circle(1pt);
\filldraw (v4) circle(1pt);
\end{scope}
\end{tikzpicture}
}
, &
\gamma_3 &= 
{
\def\cvmm {}
\def\cvmt {}
\def\cvrt {}
\def\cvrm {}
\def\cvrb {}
\def\cvlt {}
\def\cvlm {}
\def\cvlb {}
\begin{tikzpicture}[x=\scale,y=\scale,baseline={([yshift=-.5ex]current bounding box.center)}]
\begin{scope}[node distance=1]
\coordinate (i1);
\coordinate[right=.5 of i1] (v0);
\coordinate[right=.5 of v0] (v2);
\coordinate[right=.5 of v2] (vm);
\coordinate[above=of vm] (v3);
\coordinate[below=of vm] (v4);
\coordinate[right=.5 of vm] (v2d);
\coordinate[right=.5 of v2d] (v1);
\coordinate[right=.5 of v1] (o1);
\ifdefined \cvlt
\draw[line width=1.5pt] (v0) to[bend left=45] (v3);
\else
\draw (v0) to[bend left=45] (v3);
\fi
\ifdefined \cvlb
\draw[line width=1.5pt] (v0) to[bend right=45] (v4);
\else
\draw (v0) to[bend right=45] (v4);
\fi
\ifdefined \cvlm
\draw[line width=1.5pt] (v0) -- (v2);
\else
\draw (v0) -- (v2);
\fi
\ifdefined \cvrt
\draw[line width=1.5pt] (v1) to[bend right=45] (v3);
\else
\draw (v1) to[bend right=45] (v3);
\fi
\ifdefined \cvrb
\draw[line width=1.5pt] (v1) to[bend left=45] (v4);
\else
\draw (v1) to[bend left=45] (v4);
\fi
\ifdefined \cvmt
\draw[line width=1.5pt] (v3) to[bend right=20] (v2);
\else
\draw (v3) to[bend right=20] (v2);
\fi
\ifdefined \cvmb
\draw[line width=1.5pt] (v4) to[bend left=20] (v2);
\else
\draw (v4) to[bend left=20] (v2);
\fi
\ifdefined \cvmm
\draw[line width=1.5pt] (v3) to[bend left=20] (v2d);
\draw[line width=1.5pt] (v4) to[bend right=20] (v2d);
\else
\draw (v3) to[bend left=20] (v2d);
\draw (v4) to[bend right=20] (v2d);
\fi
\filldraw[color=white] (v2d) circle(.2);
\ifdefined \cvrm
\draw[line width=1.5pt] (v1) -- (v2);
\else
\draw (v1) -- (v2);
\fi
\draw (v0) -- (i1);
\draw (v1) -- (o1);
\filldraw (v0) circle(1pt);
\filldraw (v1) circle(1pt);
\filldraw (v2) circle(1pt);
\filldraw (v3) circle(1pt);
\filldraw (v4) circle(1pt);
\end{scope}
\end{tikzpicture}
}
\\
\delta_1 &= 
{
\def\cvmt {}
\def\cvmm {}
\def\cvmb {}
\def\cvrt {}
\def\cvrm {}
\def\cvrb {}
\begin{tikzpicture}[x=\scale,y=\scale,baseline={([yshift=-.5ex]current bounding box.center)}]
\begin{scope}[node distance=1]
\coordinate (i1);
\coordinate[right=.5 of i1] (v0);
\coordinate[right=.5 of v0] (v2);
\coordinate[right=.5 of v2] (vm);
\coordinate[above=of vm] (v3);
\coordinate[below=of vm] (v4);
\coordinate[right=.5 of vm] (v2d);
\coordinate[right=.5 of v2d] (v1);
\coordinate[right=.5 of v1] (o1);
\ifdefined \cvlt
\draw[line width=1.5pt] (v0) to[bend left=45] (v3);
\else
\draw (v0) to[bend left=45] (v3);
\fi
\ifdefined \cvlb
\draw[line width=1.5pt] (v0) to[bend right=45] (v4);
\else
\draw (v0) to[bend right=45] (v4);
\fi
\ifdefined \cvlm
\draw[line width=1.5pt] (v0) -- (v2);
\else
\draw (v0) -- (v2);
\fi
\ifdefined \cvrt
\draw[line width=1.5pt] (v1) to[bend right=45] (v3);
\else
\draw (v1) to[bend right=45] (v3);
\fi
\ifdefined \cvrb
\draw[line width=1.5pt] (v1) to[bend left=45] (v4);
\else
\draw (v1) to[bend left=45] (v4);
\fi
\ifdefined \cvmt
\draw[line width=1.5pt] (v3) to[bend right=20] (v2);
\else
\draw (v3) to[bend right=20] (v2);
\fi
\ifdefined \cvmb
\draw[line width=1.5pt] (v4) to[bend left=20] (v2);
\else
\draw (v4) to[bend left=20] (v2);
\fi
\ifdefined \cvmm
\draw[line width=1.5pt] (v3) to[bend left=20] (v2d);
\draw[line width=1.5pt] (v4) to[bend right=20] (v2d);
\else
\draw (v3) to[bend left=20] (v2d);
\draw (v4) to[bend right=20] (v2d);
\fi
\filldraw[color=white] (v2d) circle(.2);
\ifdefined \cvrm
\draw[line width=1.5pt] (v1) -- (v2);
\else
\draw (v1) -- (v2);
\fi
\draw (v0) -- (i1);
\draw (v1) -- (o1);
\filldraw (v0) circle(1pt);
\filldraw (v1) circle(1pt);
\filldraw (v2) circle(1pt);
\filldraw (v3) circle(1pt);
\filldraw (v4) circle(1pt);
\end{scope}
\end{tikzpicture}
}
, &
\delta_2 &= 
{
\def\cvmt {}
\def\cvmm {}
\def\cvmb {}
\def\cvlt {}
\def\cvlm {}
\def\cvlb {}
\begin{tikzpicture}[x=\scale,y=\scale,baseline={([yshift=-.5ex]current bounding box.center)}]
\begin{scope}[node distance=1]
\coordinate (i1);
\coordinate[right=.5 of i1] (v0);
\coordinate[right=.5 of v0] (v2);
\coordinate[right=.5 of v2] (vm);
\coordinate[above=of vm] (v3);
\coordinate[below=of vm] (v4);
\coordinate[right=.5 of vm] (v2d);
\coordinate[right=.5 of v2d] (v1);
\coordinate[right=.5 of v1] (o1);
\ifdefined \cvlt
\draw[line width=1.5pt] (v0) to[bend left=45] (v3);
\else
\draw (v0) to[bend left=45] (v3);
\fi
\ifdefined \cvlb
\draw[line width=1.5pt] (v0) to[bend right=45] (v4);
\else
\draw (v0) to[bend right=45] (v4);
\fi
\ifdefined \cvlm
\draw[line width=1.5pt] (v0) -- (v2);
\else
\draw (v0) -- (v2);
\fi
\ifdefined \cvrt
\draw[line width=1.5pt] (v1) to[bend right=45] (v3);
\else
\draw (v1) to[bend right=45] (v3);
\fi
\ifdefined \cvrb
\draw[line width=1.5pt] (v1) to[bend left=45] (v4);
\else
\draw (v1) to[bend left=45] (v4);
\fi
\ifdefined \cvmt
\draw[line width=1.5pt] (v3) to[bend right=20] (v2);
\else
\draw (v3) to[bend right=20] (v2);
\fi
\ifdefined \cvmb
\draw[line width=1.5pt] (v4) to[bend left=20] (v2);
\else
\draw (v4) to[bend left=20] (v2);
\fi
\ifdefined \cvmm
\draw[line width=1.5pt] (v3) to[bend left=20] (v2d);
\draw[line width=1.5pt] (v4) to[bend right=20] (v2d);
\else
\draw (v3) to[bend left=20] (v2d);
\draw (v4) to[bend right=20] (v2d);
\fi
\filldraw[color=white] (v2d) circle(.2);
\ifdefined \cvrm
\draw[line width=1.5pt] (v1) -- (v2);
\else
\draw (v1) -- (v2);
\fi
\draw (v0) -- (i1);
\draw (v1) -- (o1);
\filldraw (v0) circle(1pt);
\filldraw (v1) circle(1pt);
\filldraw (v2) circle(1pt);
\filldraw (v3) circle(1pt);
\filldraw (v4) circle(1pt);
\end{scope}
\end{tikzpicture}
}
& &
&\end{aligned}$
\\
with the complete forests $\emptyset \subset \delta_1 \subset \alpha_i \subset \Gamma$, 
$\emptyset \subset \delta_2 \subset \beta_i \subset \Gamma$ and
$\emptyset \subset \gamma_i \subset \Gamma$.
  }
\caption{Counter example of a lattice, which appears in join-meet-renormalizable QFTs with four-valent vertices and is not graded.}
\label{fig:counterexplsemimod}
\end{figure}
\else

MISSING IN DRAFT MODE

\fi
}

Not all join-meet-renormalizable theories have this property for every Feynman diagram. For instance, in $\phi^4$-theory in $4$-dimensional spacetime, the diagram depicted in figure \ref{fig:diagramnonranked} with its subdiagrams in figure \ref{fig:diagramnonrankeddetails} appears. The corresponding lattice, shown in figure \ref{fig:diagramnonrankedlattice}, is not graded.

The appearance of these diagrams with non-graded lattices is characteristic for theories with four-valent vertices. In theories with only three-or-less-valent vertices all lattices are graded:
\begin{thm}{\cite[Thm. 4]{borinsky2015}}
\label{thm:gradthree}
In a renormalizable QFT with only three-or-less-valent vertices:
\begin{itemize}
\item $\sdsubdiags_D(\Gamma)$ is a graded lattice for every propagator, vertex-type diagram or disjoint unions of both.
\item $\hopflat$ is bigraded by $\nu(\hat{1})$ and the length of the maximal chains of the lattices, which coincides with the coradical degree in $\hopflat$.
\item $\hopffg_D$ is bigraded by $h_1(\Gamma)$ and the coradical degree of $\Gamma$. 
\item Every complete forest of $\Gamma$ has the same length.
\end{itemize}
\end{thm}

In theories with four-valent vertices, we can also enforce the disappearance of all non-graded lattices by working in a renormalization scheme where \textit{tadpole}-diagrams vanish. Tadpoles are diagrams which can be separated in two connected components by the removal of a single vertex such that one connected component does not contain any external legs of the initial diagram. Tadpole diagrams are also called snail or seagull diagrams. 

If we use such a renormalization scheme, we can define a Hopf ideal $I$ generated by all tadpole diagrams of the initial Hopf algebra $\hopffg_D$ and form the quotient $\hopffgs_D: = \hopffg_D/I$. Instead of working with $\hopffg_D$ the quotient $\hopffgs_D$ can be used without changing any results, because the Feynman rules vanish on the ideal $I$ by requirement. 
In this quotient, the lattices corresponding to the Feynman diagrams behave in a similar way as for theories with only three valent vertices!
\begin{thm}{\cite[Thm. 5]{borinsky2015}}
\label{thm:gradfour}
In a renormalizable QFT with only four-or-less-valent vertices:
\begin{itemize}
\item $\sdsubdiagsn_D(\Gamma)$ is a graded lattice for every propagator, vertex-type diagram or disjoint unions of both.
\item $\hopflat/\chi_D(I)$ is bigraded by $\nu(\hat{1})$ and the length of the maximal chains of the lattices, which coincides with the coradical degree in $\hopflat$.
\item $\hopffgs_D:= \hopffg_D/I$ is bigraded by $h_1(\Gamma)$ and the coradical degree of $\Gamma$. 
\item Every complete forest of $\Gamma$, which does not result in a tadpole upon contraction, has the same length.
\end{itemize}
\end{thm}
where $\sdsubdiagsn_D(\Gamma)$ is the set of s.d.\ subdiagrams $\gamma$ of $\Gamma$ which do not yield tadpole diagrams upon contraction $\Gamma/\gamma$.

\section{The quotient $\hopffgs_D$: applications}
Kinematic renormalization schemes $\Phi_R:\hopffgs_D \to \mathbb{C}$ cover renormalization schemes which allow for well-defined asymptotic states and hence are natural from a physicists viewpoint.
Such schemes evaluate tadpole graphs to zero and hence are naturally defined for the above quotient $\hopffgs_D$ as $\Phi_R(I)=0$. 

Evaluating graphs by renormalized Feynman rules 
in such schemes leads to periods which have a motivic interpretation \cite{BEK,Bl,BingenLect}. We discuss some of such schemes most crucial aspects. We closely follow \cite{BrKr}
in this section. 
As usual we concentrate on scalar field theory which is generic for the whole situation.   

As we saw already amplitudes in quantum field theory can be written as a function of a chosen
scale variable $L=\ln(S/\mu^2)$ chosen such that it only vanishes when 
all external momenta vanish. We take $S$ to be a suitable linear combination 
of scalar products $q_i\cdot q_j$ of external momenta and squared masses $m_e^2$.
Dimensionless scattering angles $\Theta$ are defined accordingly as ratios $q_i\cdot q_j/S$ and $m_e^2/S$.
 
In these variables, amplitudes can be calculated as a perturbation expansion 
in terms of Feynman graphs $\Gamma$ as $\sum_\Gamma\Phi_R(\Gamma)$. Here, the renormalized Feynman rules $\Phi_R$ are expressed in terms of such
angle and scale variables, and the graphs $\Gamma$ are chosen in our quotient Hopf algebra $\hopffgs_D$. 

For any choice of angle and scale variables, $\Phi_R$ is in the group $\mathrm{Spec}_{\mathbb{C}}(\hopffgs_D)$, and the restriction of this group to maps which originate from evaluation of graphs by Feynman rules defines a sub-group $G_\text{Feyn}:=\mathrm{Spec}_{\mathrm{Feyn}}(\hopffgs_D)\subset\mathrm{Spec}_{\mathbb{C}}(\hopffgs_D)$.

Such a chosen decomposition of the variables reflects itself then in a chosen decomposition of the group $G_\text{Feyn}$ into two subgroups $G_{\mathrm{o.s.}}$, maps dependent on only one scale (o.s.) and $G_{\mathrm{fin}}$, maps dependent only on the angles. Elements $\Phi\in G_{\mathrm{o.s.}}$ are of the form 
\begin{equation}
\Phi(\Gamma)=\sum_{j=1}^{\mathrm{cor}(\Gamma)}p_j L^j,
\end{equation}
where the coefficients $p_j$ are periods in the sense of algebraic geometry and are independent of the angles $\{\Theta\}$, with the coradical degree $\mathrm{cor}(\Gamma)$ giving the bound.

Still following \cite{BrKr}, we allow for renormalization conditions which are defined by kinematic constraints on Green-functions: we demand that such Green functions, regarded as functions of $S$ and $\{\Theta\}$, vanish (up to a specified order) at a reference point (in $S,\{\Theta\}$-space) given by $S_0,\{\Theta_0\}$.  We implement these constraints graph by graph.  Hence renormalized Green functions as well as renormalized Feynman rules become functions of $S,S_0,\Theta,\Theta_0$. Here, $\Theta,\Theta_0$ stand for the whole set of angles in the Feynman rules, with $\Theta_0$ specifying the renormalization point. Note that minimal subtraction is not included in our set-up, renormalized Feynman rules in that scheme do not vanish on the ideal $I$ defined by tadpole graphs.

Elements $\Phi_{\mathrm{fin}}\in G_{\mathrm{fin}}$ are of the form
\begin{equation}\Phi_{\mathrm{fin}}(\Gamma)={c^\Gamma_0}(\Theta),\end{equation}
with ${c^\Gamma_0}(\Theta)$ an $L$-independent function of the angles. 

We hence obtain the decomposition of $G_\text{Feyn}$ 
as a map $\Phi^R\mapsto (\Phi_{\mathrm{fin}},\Phi_{\mathrm{o.s.}})$, 
which proceeds then by a twisted conjugation:
\begin{equation}
G_\text{Feyn}\ni \Phi_R(S,S_0,\Theta,\Theta_0)=\Phi^{-1}_{\mathrm{fin}}(\Theta_0)\star\Phi_{\mathrm{o.s.}}(S,S_0)\star\Phi_{\mathrm{fin}}(\Theta),
\end{equation}
with $\Phi_{\mathrm{fin}}(\Theta_0),\Phi_{\mathrm{fin}}(\Theta)\in G_{\mathrm{fin}}$ and $\Phi_{\mathrm{o.s.}}(S,S_0)\in G_{\mathrm{o.s.}}$. The group law $\star$ and inversion ${}^{-1}$ are defined through the Hopf algebra underlying $G_\text{Feyn}$.

\subsection{The additive group and renormalization schemes}
The most striking aspect of kinematic renormalization schemes is that they allow for an intimate connection between the additive group $\mathbb{G}_a$ and $\mathrm{Spec}(\hopffgs_D)$.
We have $\forall h\in \hopffgs_D$ \cite{BrKr,BlKr}
\begin{align}
\label{RG} 
&\Phi_R^L(h)= \Phi_R^{L_1+L_2} = m \circ (\Phi_R^{L_1}\otimes \Phi_R^{L_2}) \circ \Delta(h)=\Phi_R^{L_1}\star\Phi_R^{L_2},\,& &L=L_1+L_2.
\end{align}
Here, $L=\ln S/\mu^2$ defines the scale relative to a renormalization scale $\mu$. $\Phi_R^L:\,\hopffgs_D\to\mathbb{C}$ are renormalized Feynman rules,
and $\Phi_R^L(\Gamma)\equiv \Phi_R^L(\Gamma)(\{\Theta,\Theta_0\})$ is a function also of angles $\{\Theta\}$ and $\{\Theta_0\}$ (the latter for the renormalization point).

Note that to derive Eq.\ (\ref{RG}) and therefore the renormalization group in the context of the quotient Hopf algebra $\hopffgs_D$ only combinatorial properties of graphs and graph polynomials are needed \cite{BlKr,BrKr}. There is an intimate connection to the representation theory of the additive group $\mathbb{G}_a$ and Tannaka categories of Feynman graphs hiding between this set-up which is studied elsewhere \cite{BKY}.

\subsection{A tower  of  Hopf algebras}
The quotient Hopf algebra $\hopffgs_D$ is actually part of a tower of Hopf algebras which was defined in \cite{WalterDirk}, which we follow closely here.
We start with the quotient $\hopffgs_\infty$ of the core Hopf algebra $\hopffg_\infty$ \cite{coreDirk} of Feynman graphs, in which every union of 1PI subdiagrams is superficially divergent, by $I$, $\hopffgs_\infty = \hopffg_\infty/I$.

$\hopffg_\infty$ contains the renormalization Hopf algebra $\hopffgs_D$ itself as a quotient Hopf algebra \cite{WalterDirk,coreDirk} and similarly $\hopffgs_\infty$ contains $\hopffgs_D$.

For the structure of Green functions with respect to the Hopf algebra $\hopffgs_D$ we write $G^r(\{Q\},\{M\},\{g\};R)$ for a generic Green function, where
\begin{itemize}
 \item $r$ indicates the residue under consideration and we write $|r|$ for its number of external legs. Amongst all possible residues, there is a set
of residues provided by the free propagators and vertices of the theory. We write $\mathcal{R}$ for this set. It is in one-to-one correspondence with field monomials in a Lagrangian approach to field theory. The set of all residues is denoted by $\mathcal{A}=\mathcal{F}\cup\mathcal{R}$, which defines $\mathcal{F}$ as those residues only present through quantum corrections.
\item $\{Q\}$ is the set of external momenta $q_{e}$ subject to the condition $\sum_{e \in r} q_{e}=0$, where the sum is over the external half edges of $r$.
\item $\{M\}$ is the set of masses in the theory.
\item $\{g\}$ is the set of coupling constants specifying the theory. Below, we proceed for the case of a single coupling constant $g$, the general case posing no principal new problems.
\item $R$ indicates a chosen kinematic renormalization scheme. 
\end{itemize}
We also note that a generic Green function $G^r(\{Q\},\{M\},\{g\};R)$ has an expansion into scalar functions
\be G^r(\{Q\},\{M\},\{g\};R)= \sum_{t(r)\in S(r)} t(r) G_{t(r)}^r(\{Q\},\{M\},\{g\};R).\ee
In terms of mass dimensions ($[m^2]=2$) we have $\mathbb{N}_0\ni [t(r)]\geq 0$
and $[G_{t(r)}^r(\{Q\},\{M\},\{g\};R)]=0$.

Here, $S(r)$ is a basis set of Lorentz covariants $t(r)$ in accordance with the quantum numbers specifying the residue $r$.
For each $t(r)\in S(r)$, there is a projector $P^{t(r)}$ onto this formfactor.

For $r\in\mathcal{R}$, we can write \be G^r(\{Q\},\{M\},\{g\};R)=\Phi(r)G^r_{\Phi(r)}(\{Q\},\{M\},\{g\};R)+R^r(\{Q\},\{M\},\{g\};R),\ee
where $R^r(\{Q\},\{M\},\{g\};R)$ sums up all formfactors $t(r)$ but it only contributes through quantum corrections. $\Phi$ are the unrenormalized Feynman rules. Applied on the residue $r$, they evaluate to the tree-level amplitude $\Phi(r)$ for the vertex or edge associated to the residue $r$.

Each $G^r(\{Q\},\{M\},\{g\};R)$ can be obtained by the evaluation of a series of 1PI graphs
\begin{align}
X^r(g) & = & \One & & -& &\sum_{\res(\Gamma)= r}g^{|\Gamma|}\frac{\Gamma}{\mathrm{Sym}(\Gamma)},& &\forall r\in \mathcal{R},& &|r|=2,\label{green1}\\
X^r(g) & = & \One & & +& & \sum_{\res(\Gamma)= r}g^{|\Gamma|}\frac{\Gamma}{\mathrm{Sym}(\Gamma)},& &\forall r\in \mathcal{R},& &|r|>2,\label{green2}\\
X^r(g) & = && && &\sum_{\res(\Gamma)= r}g^{|\Gamma|}\frac{\Gamma}{\mathrm{Sym}(\Gamma)}& ,& \forall r\notin \mathcal{R},&\label{green3}
\end{align}
where we take the minus sign for $|r|=2$ and the plus sign for $|r|>2$. Furthermore, the notation $\res(\Gamma)= r$ indicates a sum over graphs with external leg structure in accordance with $r$.

We write $\Phi,\Phi_R$ for the unrenormalized and renormalized Feynman rules regarded as a map: $\hopffgs_D \to \mathbb{C}$ from the Hopf algebra to  $\mathbb{C}$.

We have
\be G^r_{t(r)}(\{Q\},\{M\},\{g\};R)=\Phi_R^{t(r)}(X^r(g))(\{Q\},\{M\},\{g\};R),\ee
where each non-empty graph is evaluated by the  renormalized Feynman rules 
\begin{align}
\Phi_R^{t(r)}(\Gamma)&:=(\id-R) \circ m \circ (S_R^\Phi \otimes P^{t(r)}\Phi P) \circ \Delta(\Gamma)\label{renFR} \\ 
S_R^\Phi(\Gamma) &:= -R \circ m \circ ( S_R^\Phi \otimes \Phi P ) \circ \Delta(\Gamma)
\end{align}
and $\Phi_R^{t(r)}(\One)=1$, and $P$ the projection into the augmentation ideal of $\hopffgs_D$, $P^{t(r)}$ the projector on the formfactor $t(r)$ and $R$ the renormalization map.

It is in the evaluation Eq.\ (\ref{renFR}) that the coproduct of the renormalization Hopf algebra appears. Combining the combinatorial Dyson--Schwinger equations (see \cite{Kock} for a recent overview of such equations) Eqs.\ (\ref{green1}, \ref{green2}, \ref{green3}) with Feynman rules and with the renormalization group Eq.\ (\ref{RG}) turns them into ordinary non-linear differential equations studied in \cite{KYUvG1,KYUvG2} which determine the physics behind quantum field theory.

The above sum over all graphs simplifies when one takes the Hochschild cohomology of the (renormalization) Hopf algebra into account:
\be 
X^r(g)=\delta_{r,\mathcal{R}}\One\pm 
\sum_{\substack{\Gamma ~\text{1PI} \\ \res(\Gamma)= r \\ \widetilde{\Delta}(\Gamma)=0}}\frac{1}{\mathrm{Sym}(\Gamma)}g^{|\Gamma|}  B_+^\Gamma(X^r(g) Q(g)),
\ee
($-$ sign for $|r|=2$, + sign for $|r|>2$, $\delta_{r,\mathcal{R}}=1$ for $r\in \mathcal{R}, 0$ else)
with $Q(g)$ being the formal series of graphs assigned to an invariant charge of the coupling $g$:
\be 
Q^r(g):=\left[\frac{X^{r}}{\prod_{e\in r}\sqrt{X^e}}\right]^{\frac{1}{|r|-2}}.
\ee
The existence of a unique invariant charge depends on the existence of suitable coideals. Although we can define an invariant charge for every residue $r \in \mathcal{R}$ with $|r| > 2$, the Slavnov-Taylor-Identities guaranty that upon evaluation with a counter-term map, they will all give the same renormalized charge. We can therefore drop the index $r$ and write $Q = Q^r$. $B_+^\gamma$ are {\it grafting operators} which are Hochschild cocycles, and the above combinatorial Dyson--Schwinger equations can be formulated in any quotient Hopf algebra. More on such equations can be found in \cite{WalterDirk,Foissy,Yeats,Olaf,Kock}.

The existence of the equation above indicates immediately that there is a natural Hopf algebra homomorphism $\eta$ from the Hopf algebra of rooted trees $\hopfrt$ by the universal property. Together with the Hopf algebra morphism $\chi_D$ to the Hopf algebra of decorated lattices, we have the following relationships:
\begin{align}
\hopfrt \xrightarrow{\eta} \hopffg_D \xrightarrow{\chi_D} \hopflat
\end{align}
The relationships of these different Hopf algebras especially the morphism given by $\eta \circ \chi_D: \hopfrt \rightarrow \hopflat$, will be subject of a future work.

Summarizing, there is a tower of quotient Hopf algebras (all of them also in a quotient obtained by dividing by $I$)
\be \hopffgs_4 \subset\hopffgs_6\cdots\subset \hopffgs_{2n} \subset \cdots \subset \hopffgs_{\mathrm{core}}=\hopffgs_\infty,\ee
obtained by restricting the coproduct to sums over graphs which are superficially divergent in $$D=4,6,\ldots,2n,\ldots,\infty$$ dimensions. 

We can make this explicit by including the spacetime into the notation for the coproduct:
\begin{align}
        &\Delta_D: \Gamma \mapsto
        \sum \limits_{ \gamma \in \sdsubdiags_D(\Gamma) }
            \gamma \otimes \Gamma/\gamma& &:& &{\hopffg_D} \rightarrow {\hopffg_D} \otimes {\hopffg_D}, 
\end{align}

Most striking is the connection to the additive group $\mathbb{G}_a$ which establishes itself here as announced previously.
We have
\be 
X^r=\One\pm \sum_{j\geq 1} h_j^r.
\ee
It follows from the above that the representation of $\mathbb{G}_a$ on the subvectorspace ${}^*\hopffgs_D$ spanned by such generators $h_i^r$ of the sub-Hopf algebras (Foissy \cite{Foissy} the appearance of such sub-Hopf algebras in great detail) defined by a combinatorial DSE has the form $L\to\exp{LN^r}$  where $N^r$ is a lower triangular matrix for each residue $r$. More on this and the resulting Tannakian structure of Feynman graphs will be given in \cite{BKY}.

Let us conclude with two remarks which follow from this set-up.
\begin{rem}
Investigating the Cutkosky rules \cite{Cutkosky} we can write fix-point equations for cut graphs and therefore fix-point equations for imaginary part of Green functions.
Indeed, following \cite{Cutkosky}, all algebraic structures needed to study the analytic properties of amplitudes can be formulated in $\hopffgs_D$, as tadpole graphs do not allow for non-trivial variations in external momenta as there is no momentum flow through them.

The 1-cocycles $B_+^\gamma$  which run a Green function can then be decomposed according to the complete $k$-particle cuts of $\gamma$ to obtain recursive equations for Green functions and their imaginary parts. Details will be given in future work (see also \cite{Cutkosky}, in particular lemma (3) in that reference).  
\end{rem}
\begin{rem}
In the quotient $\hopffgs_D$ together with its accompanying combinatorial Dyson--Schwinger equations all renormalization group effects come from a soft logarithmic breaking of conformal invariance as their are no quadratic divergences left for kinematic renormalization schemes. 
Accordingly, Dyson--Schwinger equations are determined by kinematical boundary conditions, and the equations themselves describe the dimensionless quantum corrections to dimensionful tree-level amplitudes.

Fine-tuning or hierarchy problems are hence spurious. They are a typical consequence of using either a dimensionful regulator and/or  renormalization schemes not in accordance with the equations of motion.
\end{rem}
\bibliography{literature}

\end{document}